\documentclass[final,12pt]{elsarticle}

\makeatletter
\def\ps@pprintTitle{%
 \let\@oddhead\@empty
 \let\@evenhead\@empty
 \def\@oddfoot{}%
 \let\@evenfoot\@oddfoot}
\makeatother

\usepackage{hyperref}
\usepackage{amsfonts}
\usepackage{epsfig}
\usepackage{makeidx}
\usepackage{amsmath,amssymb,amssymb,latexsym}
\usepackage{enumerate}
\usepackage{graphicx}
\usepackage{mathtools}

\DeclareMathOperator{\thin}{thin}
\DeclareMathOperator{\pthin}{pthin}

\DeclareMathOperator{\lmimw}{lmimw}
\DeclareMathOperator{\cutw}{cutw}

\newcommand{\clawh}{\mbox{claw}_h}

\newtheorem{theorem}{Theorem}

\newtheorem{proposition}[theorem]{Proposition}
\newtheorem{lemma}[theorem]{Lemma}
\newtheorem{corollary}[theorem]{Corollary}

\newtheorem{remark}{Remark}

\newenvironment{proof}{\textit{Proof.}}{\hfill $\Box$ \\}

\date{}
\begin{document}

\begin{frontmatter}

\title{On the thinness and proper thinness of a graph\fnref{Dedication}}

\author{Flavia Bonomo}
\address{Universidad de Buenos Aires. Facultad de Ciencias Exactas y Naturales. Departamento de Computaci\'on. Buenos Aires,
Argentina. / CONICET-Universidad de Buenos Aires. Instituto de
Investigaci\'on en Ciencias de la Computaci\'on (ICC). Buenos
Aires, Argentina.} \ead{fbonomo@dc.uba.ar}

\author{Diego de Estrada}
\address{Universidad de Buenos Aires. Facultad de Ciencias Exactas y Naturales. Departamento de Computaci\'on. Buenos Aires,
Argentina.}  \ead{destrada@dc.uba.ar}

\begin{abstract} Graphs with bounded thinness were defined in
2007 as a generalization of interval graphs. In this paper we
introduce the concept of proper thinness, such that graphs with
bounded proper thinness generalize proper interval graphs. We
study the complexity of problems related to the computation of
these parameters, describe the behavior of the thinness and proper
thinness under three graph operations, and relate thinness and
proper thinness to other graph invariants in the literature.
Finally, we describe a wide family of problems that can be solved
in polynomial time for graphs with bounded thinness, generalizing
for example list matrix partition problems with bounded size
matrix, and enlarge this family of problems for graphs with
bounded proper thinness, including domination problems.
\end{abstract}


\fntext[Dedication]{\textbf{For Pavol Hell on the Occasion of his
70th Birthday}}

\end{frontmatter}

\section{Introduction}\label{sec:intro}

A graph $G=(V,E)$ is \emph{$k$-thin} if there exist an ordering
$v_1, \dots , v_n$ of $V$ and a partition of $V$ into $k$ classes
$(V^1,\dots,V^k)$ such that, for each triple $(r,s,t)$ with
$r<s<t$, if $v_r$, $v_s$ belong to the same class and $v_t v_r \in
E$, then $v_t v_s \in E$. The minimum $k$ such that $G$ is
$k$-thin is called the \emph{thinness} of $G$. The thinness is
unbounded on the class of all graphs, and graphs with bounded
thinness were introduced in~\cite{M-O-R-C-thinness} as a
generalization of interval graphs, which are exactly the $1$-thin
graphs. When a representation of the graph as a $k$-thin graph is
given, for a constant value $k$, some NP-complete problems as
maximum weighted independent set and bounded coloring with fixed
number of colors can be solved in polynomial
time~\cite{M-O-R-C-thinness,B-M-O-thin-tcs}. These algorithms were
respectively applied for improving heuristics of two real-world
problems: the Frequency Assignment Problem in GSM
networks~\cite{M-O-R-C-thinness}, and the Double Traveling
Salesman Problem with Multiple Stacks~\cite{B-M-O-thin-tcs}. In
this work we propose a framework to describe a wide family of
problems that can be solved by dynamic programming techniques on
graphs with bounded thinness, when the $k$-thin representation of
the graph is given. These problems generalize for example the list
matrix partition problems for matrices of bounded
size~\cite{F-H-K-N-P-list-matrix-part}.

We introduce here the concept of \emph{proper thinness} of a graph
with the aim of generalizing proper-interval graphs: graphs that
are proper $1$-thin are exactly proper interval graphs (see
Section \ref{sec:thin} for a definition). We extend the framework
in order to solve in polynomial time by dynamic programming many
of the domination-type problems in the literature (e.g. classified
in \cite{Arne-dom}) and their weighted versions, such as
existence/minimum (weighted) independent dominating set, minimum
(weighted) total dominating set, minimum perfect dominating set
and existence/minimum (weighted) efficient dominating set, for the
class of graphs with bounded proper thinness $k$, when the proper
$k$-thin representation of the graph is given.

The organization of the paper is the following. In
Section~\ref{sec:thin} we state the main definitions and present
some basic results on thinness. In Section~\ref{sec:algo-thin}, we
study some problems related to the recognition of $k$-thin
graphs and proper $k$-thin graphs. We analyze the computational
complexity of finding a suitable vertex partition when a vertex
ordering is given, and, conversely, finding a vertex ordering when
a vertex partition is given.

In Section~\ref{sec:width} we survey the relation of thinness and
other width parameters in graphs. In Section~\ref{sec:interv} we
relate the proper thinness of interval graphs to other interval
graph invariants, as interval count and chains of nested
intervals.

In Section~\ref{sec:algo-thin-comb} we describe a wide family of
problems that can be solved in polynomial time for graphs with
bounded thinness, when the representation is given. In
Section~\ref{sec:algo-pthin-comb} we extend that family to include
dominating-like problems that can be solved in polynomial time for
graphs with bounded proper thinness.

In Section~\ref{sec:thin-and-oper} we describe the behavior of the
thinness and proper thinness under three graph operations: union,
join, and Cartesian product. The first two results allow us to
fully characterize $k$-thin graphs by forbidden induced subgraphs
within the class of cographs. The third result is used to show the
polynomiality of the $t$-rainbow domination problem for fixed $t$
on graphs with bounded thinness.

\section{Definitions and basic results}\label{sec:thin}

All graphs in this work are finite, undirected, and have no loops
or multiple edges. For all graph-theoretic notions and notation
not defined here, we refer to West~\cite{West}. Let $G$ be a
graph. Denote by $V(G)$ its vertex set, by $E(G)$ its edge set, by
$\overline G$ its complement, by $N(v)$ the neighborhood of a
vertex $v$ in $G$, by $N[v]$ the closed neighborhood
$N(v)\cup\{v\}$, and by $\overline{N}(v)$ the non-neighbors of
$v$. If $X \subseteq V(G)$, denote by $N(X)$ the set of vertices
not in $X$ having at least one neighbor in $X$. 

Denote by $G[W]$ the subgraph of $G$ induced by $W\subseteq V(G)$,
and by $G - W$ or $G \setminus W$ the graph $G[V(G)\setminus W]$.
A subgraph $H$ (not necessarily induced) of $G$ is a
\emph{spanning subgraph} if $V(H)=V(G)$.

Denote the size of a set $S$ by $|S|$. A \emph{clique} (resp.\
\emph{stable set}) is a set of pairwise adjacent (resp.\
nonadjacent) vertices. We use \emph{maximum} to mean
maximum-sized, whereas \emph{maximal} means inclusion-wise
maximal. The use of \emph{minimum} and \emph{minimal} is
analogous.

Denote by $K_n$ the graph induced by a clique of size $n$. A
\emph{claw} is the graph isomorphic to $K_{1,3}$. Let $H$ be a
graph and $t$ a natural number. The disjoint union of $t$ copies
of the graph $H$ is denoted by $tH$.

For a positive integer $r$, the \emph{$(r \times r)$-grid} is the
graph whose vertex set is $\{(i,j) : 1 \leq i, j \leq r\}$ and
whose edge set is $\{(i,j)(k,l) : |i - k| + |j - l| = 1, \mbox{
where } 1 \leq i,j,k, l \leq r \}$.

A \emph{dominating set} in a graph is a set of vertices such that
each vertex outside the set has at least one neighbor in the set. 

A \emph{coloring} of a graph is an assignment of colors to its
vertices such that any two adjacent vertices are assigned
different colors. The smallest number $t$ such that $G$ admits a
coloring with $t$ colors (a \emph{$t$-coloring}) is called the
\emph{chromatic number} of $G$ and is denoted by $\chi(G)$. A
coloring defines a partition of the vertices of the graph into
stable sets, called \emph{color classes}. List variations of the
vertex coloring problem can be found in the literature. For a
survey on that kind of related problems, see~\cite{Tz-color}. In
the \emph{list-coloring problem}, every vertex $v$ comes equipped
with a list of permitted colors $L(v)$ for it.

For a symmetric matrix $M$ over $0, 1, *$, the \emph{$M$-partition
problem} seeks a partition of the vertices of the input graph into
independent sets, cliques, or arbitrary sets, with certain pairs
of sets being required to have no edges, or to have all edges
joining them, as encoded in the matrix $M$: $M_{ii} = 1$ means the
$i$-th set is a clique, while $M_{ii} = 0$ means the $i$-th set is
a stable set; for $i\neq j$, $M_{ij} = 1$ means every vertex of
the $i$-th set is adjacent to every vertex of the $j$-th set,
while $M_{ij} = 0$ means there are no edges from the $i$-th set to
the $j$-th set. Moreover, the vertices of the input graph can be
equipped with lists, restricting the parts to which a vertex can
be placed. In that case the problem is know as a \emph{list matrix
partition problem}. Such (list) matrix partition problems
generalize (list) colorings and (list)
homomorphisms~\cite{F-H-K-N-P-list-matrix-part}.


When discussing about algorithms and data structures, we denote by
$n$ the number of vertices of the input graph $G$.

Given a graph $G$, a weight function $w$ on $V(G)$, and a subset
$S \subseteq V(G)$, the weight of $S$, denoted by $w(S)$ is
defined as $\sum_{v \in S} w(v)$.

A class of graphs is \emph{hereditary} when if a graph $G$ is in
the class, then every induced subgraph of $G$ is in the same
class.

A graph is a \emph{cograph} if it contains no induced path of
length four.

 A graph $G(V,E)$ is a \emph{comparability graph} if
there exists an ordering $v_1, \dots , v_n$ of $V$ such that, for
each triple $(r,s,t)$ with $r<s<t$, if $v_r v_s$ and $v_s v_t$ are
edges of $G$, then so is $v_r v_t$.\ Such an ordering is a
\emph{comparability ordering}. A graph is a \emph{co-comparability
graph} if its complement is a comparability graph.

A graph $G$ is an \emph{interval graph} if to each vertex $v \in
V(G)$, can be associated a closed interval $I_v=[l_v,r_v]$ of the real line, such that two distinct vertices $u,v \in V(G)$ are adjacent if and only if $I_u \cap I_v \neq \emptyset$. The family
$\{I_v\}_{v \in V(G)}$ is an \emph{interval representation} of
$G$. An undirected graph $G$ is a \emph{proper interval graph} if
there is an interval representation of $G$ in which no interval
properly contains another. In the same way, an undirected graph
$G$ is a \emph{unit interval graph} if there is an interval
representation of $G$ in which all the intervals have the same
length.

In 1969, Roberts~\cite{Rob-uig} proved that the classes of proper
interval graphs, unit interval graphs, and interval graphs with no
claw as induced subgraph coincide.

The right-end ordering of the vertices of an interval graph
satisfies the following property: for each triple $(r, s, t)$ with
$r <s<t$, if $v_tv_r \in E$, then $v_tv_s \in E$. In other words,
the neighbors of vertex $t$ with index less than $t$ are $t-1,
t-2, \dots, t-d$ for some $d \geq 0$. Moreover, a graph $G$ is an
interval graph if and only if there exists an ordering of its
vertices satisfying the property above
\cite{R-PR-interval,Ola-interval}.

Let us repeat and extend the definition of $k$-thinness given in
the introduction. A graph $G=(V,E)$ is \emph{$k$-thin} if there
exist an ordering $v_1, \dots , v_n$ of $V$ and a partition of $V$
into $k$ classes such that, for each triple $(r, s, t)$ with $r
<s<t$, if $v_r$, $v_s$ belong to the same class and $v_tv_r \in
E$, then $v_tv_s \in E$. An ordering and a partition satisfying
those properties are said to be \emph{consistent}. The minimum $k$
such that $G$ is $k$-thin is called the \emph{thinness} of $G$ and
denoted by $\thin(G)$.

The \emph{thinness} of a graph was introduced by Mannino, Oriolo,
Ricci, and Chandran in 2007 \cite{M-O-R-C-thinness}. Graphs with
\emph{bounded thinness} (thinness bounded by a constant value) are
a generalization of interval graphs, that are exactly the graphs
with thinness 1, and capture some of their algorithmic properties.

Let $\overline{tK_2}$ be the complement of a matching of size~$t$.

\begin{theorem}\emph{\cite{M-O-R-C-thinness}}\label{thm:tK2} For every $t \geq
1$, $\thin(\overline{tK_2})=t$. \end{theorem}


The right-end ordering of the vertices of a proper interval graph
satisfies the following property: for each triple $(r, s, t)$ with
{$r <s<t$}, if {$v_tv_r \in E$}, then {$v_tv_s \in E$} and
{$v_rv_s \in E$}. In other words, the neighbors of vertex {$t$}
with index less than $t$ are {$t-1, t-2, \dots, t-d$}, and those
with index greater than $t$ are {$t+1, t+2, \dots, t+d'$}.
Moreover, $G$ is a proper interval graph if and only if there
exists an ordering of its vertices satisfying the property
above~\cite{D-H-H-circ-arc,L-O-prop-int}.

We define the concept of \emph{proper thinness} of graphs as
follows.

A graph $G=(V,E)$ is \emph{proper $k$-thin} if there exist an
ordering $v_1, \dots , v_n$ of $V$ and a partition of $V$ into $k$
classes $(V^1,\dots,V^k)$ such that, for each triple $(r,s,t)$
with $r<s<t$, if $v_r$, $v_s$ belong to the same class and $v_t
v_r \in E$, then $v_t v_s \in E$ and if $v_s$, $v_t$ belong to the
same class and $v_r v_t \in E$, then $v_r v_s \in E$.
Equivalently, $G$ is {proper $k$-thin} if {both $v_1, \dots , v_n$
and $v_n, \dots , v_1$} are consistent with the partition. In this
case, the partition and the ordering $v_1,\dots,v_n$ are said to
be \emph{strongly consistent}, and the minimum $k$ such that $G$
is proper $k$-thin is called the \emph{proper thinness} of $G$ and
denoted by {$\pthin(G)$}.

Since $k$-thin graphs are defined as a generalization of interval
graphs, proper $k$-thin graphs arise naturally as a generalization
of proper interval graphs. It can be seen that a graph is proper
$1$-thin if and only if it is a proper interval graph. Moreover,
the proper-thinness of the class of interval graphs is unbounded
(See Proposition~\ref{prop:clawh}).


\bigskip


\section{Algorithmic aspects}\label{sec:algo-thin}

\bigskip

We will deal in this section with some questions related to the
recognition problem of (proper) $k$-thin graphs. The recognition
problem itself is open so far for both classes, but we will show
that, given a vertex ordering of a graph, we can find in
polynomial time a partition into a minimum number of classes which
is (strongly) consistent with the ordering. On the other hand, we
will show that given a graph and a vertex partition, it is
NP-complete to decide if there exists an ordering of the vertices
of the graph which is (strongly) consistent with the partition.

\begin{theorem}\label{thm:thin-comp-order} Given a graph $G$ and an ordering ${<}$ of its vertices, one can find in polynomial time graphs $G_<$ and $\tilde{G}_<$ with the following properties:
    \begin{enumerate}[(1)]
\item $V(G_<)=V(\tilde{G}_<)=V(G)$;

\item\label{item:2} the chromatic number of $G_<$ (resp.
$\tilde{G}_<$) is equal to the minimum integer $k$ such that there
is a partition of $V(G)$ into $k$ sets that is consistent (resp.
strongly consistent) with the order $<$, and the color classes of
a valid coloring of $G_{<}$ (resp. $\tilde{G}_<$) form a partition
consistent (resp. strongly consistent) with ${<}$;

\item $G_<$ and $\tilde{G}_<$ are co-comparability graphs.
\end{enumerate}
In particular, the minimum integer $k$ as in (\ref{item:2}) and a
partition into $k$ vertex sets can be computed in polynomial time.
Moreover, if $G$ is a co-comparability graph and ${<}$ a
comparability ordering of $\overline{G}$, then $G_{<}$ and
$\tilde{G}_<$ are spanning subgraphs of $G$.
\end{theorem}

\begin{proof}
Let $G$ be a graph and ${<}$ an ordering of its vertices. We will
build a graph $G_{<}$ such that $V(G_{<})=V(G)$, and $v < w$ are
adjacent in $G_{<}$ if and only if they cannot belong to the same
class of a partition which is consistent with $<$. By definition
of consistency, this happens if and only if there is a vertex $z$
in $G$ such that $v < w < z$, $z$ is adjacent to $v$ and
nonadjacent to $w$. So define $E(G_{<})$ such that for $v < w$,
$vw \in E(G_{<})$ if and only if there is a vertex $z$ in $G$ such
that $v < w < z$, $zv \in E(G)$ and $zw \not \in E(G)$.

We build $\tilde{G}_<$ in a similar way. In this case, for $v <
w$, $vw \in E(\tilde{G}_{<})$ if and only if either there is a
vertex $z$ in $G$ such that $v < w < z$, $zv \in E(G)$ and $zw
\not \in E(G)$ or there is a vertex $x$ in $G$ such that $x < v <
w$, $xw \in E(G)$ and $xv \not \in E(G)$.

Let us see that $<$ is a comparability ordering both for
$\overline{G_{<}}$ and $\overline{\tilde{G}_{<}}$. Suppose on the
contrary that there is a triple $r < s < t$ in $V(G)$ such that
$rs$, $st$ are edges of $\overline{G_{<}}$ (resp.
$\overline{\tilde{G}_{<}}$) and $rt$ is not an edge of
$\overline{G_{<}}$ (resp. $\overline{\tilde{G}_{<}}$). By
definition of $G_{<}$ (resp. $\tilde{G}_{<}$), there is a vertex
$z$ such that $r < s < t < z$, $zr \in E(G)$ and $zt \not \in
E(G)$ (resp. either there is a vertex $z$ such that $r < s < t <
z$, $zr \in E(G)$ and $zt \not \in E(G)$, or there is a vertex $x$
in $G$ such that $x < r < s < t$, $xt \in E(G)$ and $xr \not \in
E(G)$). If $zs \not \in E(G)$, then $rs$ is an edge of $G_{<}$
(resp. $\tilde{G}_{<}$), a contradiction. If $zs \in E(G)$, then
$st$ is an edge of $G_{<}$ (resp. $\tilde{G}_{<}$), a
contradiction as well. The case of $x$ for $\tilde{G}_{<}$ is
symmetric, if $xs \not \in E(G)$, then $st$ is an edge of
$\tilde{G}_{<}$, a contradiction. If $xs \in E(G)$, then $rs$ is
an edge of $\tilde{G}_{<}$, a contradiction as well. So $G_{<}$
and $\tilde{G}_{<}$ are co-comparability graphs, being $<$ a
comparability ordering for $\overline{G_{<}}$ and
$\overline{\tilde{G}_{<}}$, respectively.

As we have defined $G_{<}$ (resp. $\tilde{G}_<$) such that
$V(G_{<})=V(\tilde{G}_{<})=V(G)$, and $v < w$ are adjacent in
$G_{<}$ (resp. $\tilde{G}_<$) if and only if they cannot belong to
the same class of a partition which is consistent (resp. strongly
consistent) with $<$, it follows that there is a one-to-one
correspondence between partitions of $V(G)$ consistent (resp.
strongly consistent) with $<$ and colorings of $G_{<}$ (resp.
$\tilde{G}_<$). In particular, the minimum $k$ such that there is
a partition of $V(G)$ into $k$ sets that is consistent (resp.
strongly consistent) with $<$ is the chromatic number of $G_{<}$
(resp. $\tilde{G}_<$). An optimum coloring of $G_{<}$ (resp.
$\tilde{G}_<$) can be computed in polynomial time \cite{Go-comp2}.

To complete the proof of the theorem, suppose now that $G$ is a
co-comparability graph and $<$ is a comparability ordering for
$\overline{G}$. Let $v < w$ adjacent in $G_{<}$ (resp.
$\tilde{G}_<$). By definition, there is a vertex $z$ in $G$ such
that $v < w < z$, $vz \in E(G)$ and $wz \not \in E(G)$ (resp.
either there is a vertex $z$ in $G$ such that $v < w < z$, $vz \in
E(G)$ and $wz \not \in E(G)$, or there is a vertex $x$ in $G$ such
that $x < v < w$, $xw \in E(G)$ and $xv \not \in E(G)$). If $vw
\not \in E(G)$, being $\overline{G}$ a comparability graph, $vz
\not \in E(G)$, a contradiction. So $vw \in E(G)$. This proves
that $G_{<}$ is a spanning subgraph of $G$. The case of $x$ for
$\tilde{G}_{<}$ is symmetric, if $vw \not \in E(G)$, being
$\overline{G}$ a comparability graph, $xw \not \in E(G)$, a
contradiction. So in any case $vw \in E(G)$. This proves that
$\tilde{G}_{<}$ is a spanning subgraph of $G$ as well.
\end{proof}

A direct consequence of this result is the following, that was
already proved in \cite{B-M-O-thin-tcs} for the case of thinness.

\begin{corollary}\label{cor:chrom} If $G$ is a co-comparability graph, $\thin(G) \leq \pthin(G) \leq \chi(G)$.
Moreover, any vertex partition given by a coloring of $G$ and any
comparability ordering for its complement are strongly consistent.
\end{corollary}

As already observed in~\cite{B-M-O-thin-tcs}, the bound $\thin(G)
\leq \pthin(G) \leq \chi(G)$ for co-comparability graphs can be
arbitrarily bad: for example, if $G$ is a clique of size $n$, then
$\thin(G)= \pthin(G) = 1$ and $\chi(G) = n$. However, it holds
with equality for graphs $\overline{tK_2}$, because
$\thin(\overline{tK_2})=\pthin(\overline{tK_2})=\chi(\overline{tK_2})=t$
(Theorem~\ref{thm:tK2} and Corollary~\ref{cor:chrom}). \\

In contrast with Theorem \ref{thm:thin-comp-order}, if a partition
is given, it is NP-complete to decide the existence of a
(strongly) consistent ordering. \\

\noindent \textsc{(Strongly) Consistent ordering with a given partition}\\
\emph{Instance:} A graph $G=(V,E)$ and a partition of $V$ into
non-empty subsets.\\
\emph{Question:} Does there exist a total order $<$ of $V$
(strongly) consistent with the given partition? \\

The proof is based on a reduction from the following problem,
which is known to be NP-complete \cite{Guttmann2006}. \\

\noindent \textsc{Non-Betweenness}\\
\emph{Instance:} A finite set $A$ and a collection $S$ of ordered triples of distinct elements of $A$.\\
\emph{Question:} Does there exist a total order $<$ of $A$ such
that for each $(x,y,z)\in S$, it is never the case that $x<y<z$ or
$z<y<x$ (i.e. $y$ is \emph{not between} $x$ and $z$)? \\

We start with an easy lemma.

\begin{lemma}\label{lem:order-match}
Let $G$ be a graph, $<$ an ordering of $V(G)$ and $V_1,\dots,V_k$
a partition of $V(G)$ that is consistent with $<$. Let
$\{x_i,y_i\} \subseteq V_i$, for $i=1,2$, such that $x_1x_2$ and
$y_1y_2$ are the only edges between $\{x_1,y_1\}$ and
$\{x_2,y_2\}$. Then $x_1 < y_1$ if and only if $x_2 < y_2$.
\end{lemma}

\begin{proof}
By symmetry, let us assume that $y_1$ is the biggest vertex
according to $<$. Again by symmetry, to prove the lemma it is
enough to prove that $x_2 < y_2$. By definition of
consistency, since $x_2$ and $y_2$ are in the same class and $y_1$
is adjacent to $y_2$ but not to $x_2$, it is not possible that
$y_2 < x_2 < y_1$.
\end{proof}

\begin{theorem}\label{np-c-order} The problem \textsc{(Strongly) Consistent ordering with a given partition} is NP-complete. \end{theorem}

\begin{proof}
First note that \textsc{(Strongly) Consistent ordering with a
given partition} is in NP, by using the total order of $V$ as the
certificate.

Now let us prove its NP-hardness. Given an instance $(A,S)$ of
\textsc{Non-Betweenness}, build a graph $G=(V,E)$ and a partition
$V_0,V_1,\dots V_{|S|}$ of $V$ as follows.

Fix an ordering of the triples in $S$. Vertices of $V_0$ are in one-to-one correspondence with elements of $A$. For $i=1,\dots,|S|$, $V_i$ has $3$ vertices, and they are in a one-to-one correspondence with the elements of the $i$-th triple in $S$. Let us call
$a^i$ the element of $V_i$ that corresponds to $a \in A$, for
$i=0, \dots, |S|$.

Define the edges
of $G$ as follows: for each triple $(x,y,z)\in S$, let $V_i$ be
its corresponding set. The only edge in the subgraph induced by
$\{x^i,y^i,z^i\}$ is $x^iz^i$. The remaining edges of $G$ are all
the possible edges between vertices associated to the same $a\in
A$.

Suppose first there is an ordering $<$ consistent with the
partition $\{V_0,$ $\dots,$ $V_{|S|}\}$.  By
Lemma~\ref{lem:order-match}, for each $1 \leq i \leq |S|$, the
relative order of the vertices $x^i, y^i, z^i$ is the same as the
relative order of the vertices $x^0, y^0, z^0$. By definition of
consistency and since the only edge in the subgraph induced by
$\{x^i,y^i,z^i\}$ is $x^iz^i$, $y^i$ is not between $x^i$ and
$z^i$ in that order. So the order of the vertices in $V_0$ gives a
positive answer to the instance $(A,S)$ of
\textsc{Non-Betweenness}.

Suppose now that there is a valid order $<$ for the instance
$(A,S)$ of \textsc{Non-Betweenness}. We can extend $<$ to $V(G)$
by making consecutive all the copies in $V(G)$ of an element of
$A$. Now, let $p<q<r$ be three vertices of $G$ such that $p$, $q$
belong to the same class $V_i$ and $rp \in E(G)$. Since $V_0$ is a
stable set and the triples in $S$ satisfy the non-betweenness
condition, $r$ is not in $V_i$. So $r$ and $p$ correspond to the
same element $a$ of $A$, and since there is at most one copy of an
element of $A$ in each $V_i$, $q$ does not correspond to a copy of
$a$. But this contradicts the fact that all the vertices of $G$
that correspond to a same element of $A$ are consecutive. So the
situation cannot arise, and the extended order is consistent with
the partition. The case in which $q$, $r$ belong to the same class
$V_i$ is identical, and indeed the extended order is strongly
consistent with the partition.
\end{proof}

The computational complexity of the decision of existence of a
(strongly) consistent ordering when the number of sets in the
partition is fixed is still open. So is the computational
complexity of deciding if a graph is (proper) $k$-thin, even for
fixed $k \geq 2$. In the case of proper thinness, the problem is
open even within the class of interval graphs.

\section{Thinness and other width parameters}\label{sec:width}

\bigskip

Many width parameters are defined in the literature. In this
section we compile the results relating the thinness with some of
them, namely pathwidth~\cite{R-S-minors1-pw},
treewidth~\cite{B-B-treewidth,R-S-minors2-tw},
clique-width~\cite{Cour-cw}, cutwidth~\cite{Lenga-cutwidth},
mim-width~\cite{VatshelleThesis}, and boxicity~\cite{Rob-box}.

In~\cite{M-O-R-C-thinness} it was proved that the thinness of a graph is at most the pathwidth plus one, and that the gap may be high, since the pathwidth of a complete graph with $r$ vertices is $r-1$, while its thinness is $1$.

On the other hand, in~\cite{C-M-O-thinness-man} it was proved that
the boxicity is a lower bound for the thinness of a graph, and it
was pointed out that the difference can be large, as an $(r\times
r)$-grid has boxicity $2$ and thinness $\Theta(r)$.

The \emph{vertex isoperimetric peak} of a graph $G$, denoted as
$b_v(G)$, is defined as $\max_s \min_{X\subset V, |X|=s} |N(X)|$.
The thinness of the grid was estimated by using the following
result, that was also used in~\cite{B-C-thinness} to give a lower
bound of  the thinness of a complete binary tree. We will use
it as well to estimate the thinness of complete $m$-ary trees.

\begin{lemma}\emph{\cite{C-M-O-thinness-man}}\label{lem:peak}
For every graph $G$, $\thin(G)\geq b_v(G)/\Delta(G)$.
\end{lemma}

Interval graphs have thinness $1$ and unbounded
clique-width~\cite{G-R-clique-width}, while cographs have
clique-width $2$~\cite{C-O-cw-cographs} and unbounded thinness,
because $\overline{tK_2}$ is a cograph for every $t$, so the
parameters are not comparable.

Complete graphs have high treewidth and thinness $1$, and
trees instead have treewidth $1$ but we have the following result.

\begin{theorem}\label{thm:thin-tree} For every fixed value $m$, the thinness of the complete $m$-ary tree on $n$ vertices is
$\Theta(\log n)$. \end{theorem}

\begin{proof}
In~\cite{Vrto-isop} it was proved that the vertex isoperimetric peak of the complete $m$-ary tree of height $h$ is $\Theta(h)$.
On the other hand, it was proved in~\cite{E-S-T-pw-tree,Sch-pw-tree} that the pathwidth of the complete $m$-ary tree of height $h$ is $\Theta(h)$. As the thinness of
a graph is upper bounded by the pathwidth plus one~\cite{M-O-R-C-thinness} and using Lemma~\ref{lem:peak}, it follows that the thinness of the complete $m$-ary tree of height $h$ is $\Theta(h)$, and this proves the theorem.
\end{proof}

The \emph{cutwidth} of a graph $G$, denoted as $\cutw(G)$,
is the smallest integer $k$ such that the vertices of $G$ can be
arranged in a linear layout $v_1,\dots,v_n$ in such a way that for
every $i=1,\dots,n-1$, there are at most $k$ edges with one
endpoint in $\{v_1,\dots,v_i\}$ and the other in $\{v_{i+1},
\ldots, v_n\}$.

\begin{theorem}\label{thm:thin-cutw} For every graph $G$, $\thin(G) \leq \cutw(G)+1$.
Moreover, a linear layout realizing the cutwidth leads to a
consistent partition into at most $\cutw(G)+1$ classes.
\end{theorem}

\begin{proof}
Let $G$ be a graph of cutwidth $k$, and let $v_1,\dots,v_n$ such
that for every $i=1,\dots,n-1$, there are at most $k$ edges with
one endpoint in $\{v_1,\dots,v_i\}$ and the other in $\{v_{i+1},
\ldots, v_n\}$. Let $G_<$ be the graph defined as in Theorem
\ref{thm:thin-comp-order} for the order $v_1,\dots,v_n$. Since
$G_<$ is a co-comparability graph, its chromatic number equals the
size of a maximum clique of it~\cite{Meyn-co-comp}. Suppose that
$G_<$ has a clique $H$ of size $k+2$, and let $v_i$ be the vertex
of higher index in $H$. By definition of $G_<$, for each $i' < i$
such that $v_{i'} \in H$, there exists $j
> i$ such that $v_j$ is adjacent to $v_{i'}$ and not adjacent to
$v_i$. So, there are at least $k+1$ edges with one endpoint in
$\{v_1,\dots,v_i\}$ and the other in $\{v_{i+1}, \ldots, v_n\}$, a
contradiction.
\end{proof}

The gap may be high, as for example on cliques.

The \emph{linear MIM-width} of a graph $G$, denoted as
$\lmimw(G)$, is the smallest integer $k$ such that the
vertices of $G$ can be arranged in a linear layout $v_1,\dots,v_n$
in such a way that for every $i=1,\dots,n-1$, the size of a
maximum induced matching in the bipartite graph formed by the
edges of $G$ with an endpoint in $\{v_1,\dots,v_i\}$ and the other
one in $\{v_{i+1}, \ldots, v_n\}$ is at most $k$. This is the
linear version of a parameter called
MIM-width~\cite{VatshelleThesis}, that is a lower bound for the
linear MIM-width.

\begin{theorem}\label{thm:thin-mimw} For every graph $G$, $\lmimw(G) \leq \thin(G)$. Moreover, a linear ordering $v_1, \dots, v_n$ realizing the
thinness, satisfies that the size of a maximum induced matching in
the bipartite graph formed by the edges of $G$ with an endpoint in
$\{v_1,\dots,v_i\}$ and the other one in $\{v_{i+1}, \ldots,
v_n\}$ is at most $\thin(G)$. \end{theorem}

\begin{proof}
Let $k=\thin(G)$ and consider a $k$-thin representation of $G$,
with ordering $<$ of $V(G)$, namely $v_1 < \dots < v_n$, and a
partition of $V(G)$ into $k$ classes. Let $1 \leq i \leq n-1$ and
let $M$ be a maximum induced matching in the bipartite graph
formed by the edges of $G$ with an endpoint in $\{v_1,\dots,v_i\}$
and the other one in $\{v_{i+1}, \ldots, v_n\}$. Suppose $v_rv_t$
and $v_sv_q$ belong to $M$, with $r < s \leq i$, $t, q \geq i+1$.
If $v_r$ and $v_s$ belong to the same class of the partition, by
definition of $k$-thin representation, $v_sv_t$ is also an edge, a
contradiction with the fact that $M$ is an induced matching. So,
$|M| \leq k$, thus $\lmimw(G) \leq \thin(G)$.
\end{proof}

As a corollary, given a graph $G$ provided with a $k$-thin
representation, a wide family of problems known as Locally
Checkable Vertex Subset and Vertex Partitioning Problems (LC-VSVP
Problems) can be solved in $n^{(O(k))}$ time~\cite{VatshelleThesis},
as this holds for MIM-width $k$ and a suitable ordering. This family of
problems is not comparable (inclusion-wise) with the one in
Section~\ref{sec:algo-thin-comb}, but encompasses maximum weighted
independent set and minimum weighted dominating set.

\section{Interval graphs with high proper thinness}\label{sec:interv}

In this section we first show that proper thinness of the class of
interval graphs is unbounded. Then we relate the proper thinness
of interval graphs to other interval graphs invariants, like
interval count. A family of interval graphs with arbitrarily large
proper thinness is the following.

Let $h \geq 1$, and define \emph{claw$_h$} as the graph obtained
from the complete ternary tree of height $h$ by adding all the
edges between a vertex of the tree and its ancestors. It is easy
to see that $\clawh$ is an interval graph for every $h \geq 1$ (an
interval representation of claw$_3$ can be seen in
Figure~\ref{fig:clawh}). The graph claw$_1$ is the claw.

\begin{proposition}\emph{\cite{Saban2010-pc}}\label{prop:clawh}
For any $h \geq 1$, $\pthin(\clawh)=h+1$.
\end{proposition}

\begin{proof}
Let $h \geq 1$. We will label the vertices of $G = \clawh$ as
$v^i_j$ such that $0 \leq i \leq h$, $1 \leq j \leq 3^i$, $v^0_1$
is the root of the ternary tree, and for each $0 \leq i \leq h-1$,
$1 \leq j \leq 3^i$, the children of $v^i_j$ are $v^{i+1}_{3j-2}$,
$v^{i+1}_{3j-1}$, and $v^{i+1}_{3j}$. Let us consider an ordering
$<$ and a partition of $V(G)$ that are strongly consistent.
Without loss of generality, by symmetry, we may assume $v^i_2 <
v^i_1 < v^i_3$ for every $i \geq 1$.

Let us show now that for every $0 \leq i' < i \leq h$, $v^i_1$ and
$v^{i'}_1$ cannot be in the same class of the partition.
Otherwise, if $v^i_1 < v^{i'}_1$ then the fact of $v^i_2 < v^i_1 <
v^{i'}_1$, $v^i_2 v^{i'}_1 \in E(G)$ and $v^i_2 v^i_1  \not \in
V(G)$ contradicts the definition of strong consistency, and if
$v^{i'}_1 < v^i_1$ then the fact of $v^{i'}_1 < v^i_1 < v^i_3$,
$v^i_3 v^{i'}_1 \in E(G)$ and $v^i_3 v^i_1 \not \in V(G)$
contradicts the definition of strong consistency.

So, $v^0_1, \dots, v^h_1$ are all in different classes of the
partition and $\pthin(\clawh)\geq h+1$. On the other hand, a
partition of the vertices according to its height in the tree, and
a postorder of the vertices of the tree are strongly consistent.
Thus $\pthin(\clawh)= h+1$.
\end{proof}

\begin{figure}
\begin{center}
\includegraphics[width=.7\textwidth]{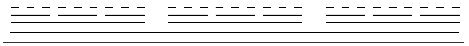}
\caption{An interval representation of claw$_3$.}\label{fig:clawh}
\end{center}
\end{figure}

This example is also a classical example of a graph with high
interval count and high length of a chain of nested intervals. We
will relate the proper thinness of interval graphs to these two
interval graphs invariants.

The \emph{interval count} of an interval graph $G$ is the minimum
number of different interval sizes needed in an interval
representation of $G$ (see for example
\cite{C-O-S-count,L-A-P-count}). Graphs with interval count at
most $k$ are also known as \emph{$k$-length interval graphs}.

A \emph{$k$-nested interval graph} is an interval graph admitting
an interval representation in which there are no chains of $k + 1$
intervals nested in each other~\cite{K-O-S-nested}. It is easy to
see that $k$-nested interval graphs are a superclass of $k$-length
interval graphs. We have also the following property.

\begin{proposition}\emph{\cite{Nisse2010-pc}}
Every $k$-nested interval graph is proper $k$-thin.
\end{proposition}

\begin{proof}
Let $G$ be a $k$-nested interval graph and consider an interval
representation of $G$ with no chains of $k + 1$ intervals nested
in each other. It is a known result that we may assume that all
the interval endpoints are distinct. We label each interval by the
length of the longest chain of nested intervals ending in it, and
these labels define the partition of the vertices into classes,
that are at most $k$. Now, we order the vertices according to
their intervals by the right endpoint (left to right). That order
is consistent with the partition in which the only class contains
all vertices of $G$, so, in particular, it is consistent with
every other partition refining it. Let us see that the consistency
is strong. Let $r < s < t$ such that $s$ and $t$ are in the same
class of the partition. Let $I_r, I_s, I_t$ their corresponding
intervals. By definition of the classes, $I_s \not \subseteq I_t$,
otherwise the length of the longest chain of nested intervals
ending in $I_s$ would be strictly greater than the one for $I_t$.
As the right endpoint of $I_t$ is greater than the one of $I_s$,
it follows that the left endpoint of $I_t$ is also greater than
the one of $I_s$. Thus, if $I_r$ intersects $I_t$, it intersects
$I_s$ as well. So, the ordering and the partition are strongly
consistent and $G$ is proper $k$-thin.
\end{proof}

Graphs with interval count one are known as unit interval graphs,
while $1$-nested interval graph are equivalent to proper interval
graphs. In~\cite{Rob-uig} it is shown that unit interval graphs
are equivalent to proper interval graphs. So the classes proper
$1$-thin, $1$-length interval and $1$-nested interval are
equivalent. We will see that for higher numbers the equivalence
does not necessarily hold.

Indeed, in~\cite[Theorem 5, p. 177]{Fishburn-book}, Fishburn shows
that, for every $k \geq 2$, there are $2$-nested interval graphs
that are not $k$-length interval.

We will describe a family of graphs that show that, for every $k
\geq 3$, there are proper $3$-thin graphs that are not $k$-nested
interval.

Let $k \geq 1$. Let $G_k$ with $3k+1$ vertices is defined as
follows. Its vertex-set is $V_k = A_k  \cup B_k \cup W_k$, where
$A_k = \{a_1, \dots, a_k\}$, $B_k = \{b_1, \dots, b_k\}$ and $W_k
= \{v_1, \dots, v_k, v_{k+1}\}$. The subgraph induced by $W_k$ is
a clique with $k+1$ vertices; $a_1$ (resp., $b_1$) is adjacent to
$v_1$. Then, for any $1 < i \leq k$, $a_i$ (resp., $b_i$) is
adjacent to $a_{i-1}$ (resp., to $b_{i-1}$), and to $v_j$ for any
$j \geq i$. See Figure~\ref{fig:Nisse} for a sketch of $G_k$ and
an interval representation of it.

The graph $G_1$ is the claw, which is not proper interval. For
higher values of $k$, we have the following property.

\begin{figure}
\begin{center}
\includegraphics[width=.9\textwidth]{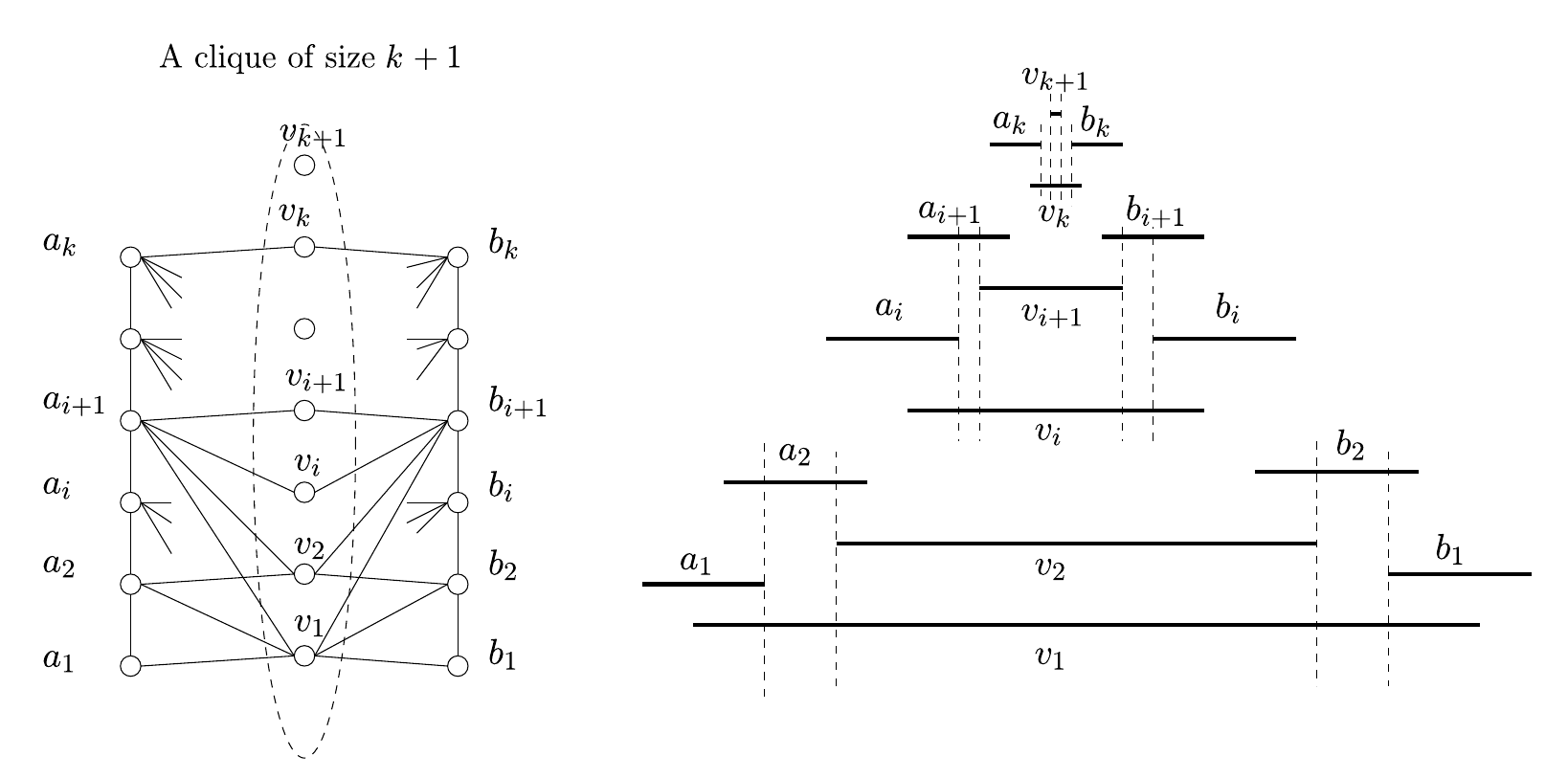}
\caption{A sketch of graph $G_k$ and an interval representation of
it.}\label{fig:Nisse}
\end{center}
\end{figure}

\begin{proposition}\emph{\cite{Nisse2010-pc}}
For any $k \leq 2$, $G_k$ is proper $3$-thin, but in every
interval representation of it, if $I_j$ is the interval
corresponding to $v_j$, it holds $I_{k+1} \subseteq I_k \subseteq
\dots \subseteq I_1$.
\end{proposition}

\begin{proof}
Consider the ordering $a_1, \dots, a_k, b_1, \dots, b_k, v_1,
\dots, v_k, v_{k+1}$, and the three classes $A_k$, $B_k$ and
$W_k$. It is easy to see that they are strongly consistent.

Let $1 \leq i \leq k-1$. Notice that $a_i a_{i+1} v_{i+1} b_{i+1}
b_i$ induce a path of length five on $G_k$. In every interval
representation of it, the interval $I_{i+1}$ is between the
intervals corresponding to $a_i$ and $b_i$ and disjoint to them.
As the five vertices are adjacent to $v_i$, it follows that the
$I_{i+1} \subseteq I_{i}$. Finally, by the shape of interval
representations of a path of length five, each of the intervals
corresponding to $a_k$ and $b_k$ contains an endpoint of $I_k$. As
$v_{k+1}$ is neither adjacent to $a_k$ nor to $b_k$, $I_{k+1}
\subseteq I_k$.
\end{proof}

The following characterization was proved for $k$-nested interval
graphs.

\begin{lemma}\emph{\cite{K-O-S-nested}} An interval graph is $k$-nested interval if and only if it
has an interval representation which can be partitioned into $k$
proper interval representations.
\end{lemma}

This lemma and the family of graphs $G_k$ show that even if the
vertices of a proper $k$-thin graph can be partitioned into $k$
sets of vertices each of them inducing a proper interval graph, it
is not always the case that it has an interval representation
which can be partitioned into $k$ proper interval representations.

\section{Solving combinatorial optimization problems on graphs
with bounded thinness}\label{sec:algo-thin-comb}

\bigskip

Since a $k$-thin graph $G$ does not contain $\overline{(k+1)K_2}$
as induced subgraph (Theorem~\ref{thm:tK2}), it has {at most
$|V(G)|^{2k}$ maximal cliques}~\cite{Pr-few-cliques}. In
particular, the maximum weighted clique problem can be solved in
polynomial time on graphs with bounded thinness, by simple
enumeration of the maximal cliques of the
graph~\cite{T-I-A-S-cliques}.

The maximum weighted stable set problem can be solved in
polynomial time on graphs with bounded thinness, when an ordering
and a partition that are consistent are
given~\cite{M-O-R-C-thinness}.
 In the same hypothesis, the capacitated
coloring (in which there is an upper bound $\alpha_j$ on the
number of vertices of color $j$) can be solved in polynomial time,
if the number of colors $s$ is fixed~\cite{B-M-O-thin-tcs}. As a
byproduct, in the same paper it is shown that the capacitated
coloring can be solved in polynomial time for co-comparability
graphs, if the number of colors $s$ is fixed, in contrast with the
case in which the bounds $\alpha_j$ are all equal to a fixed
number $h$, that is NP-complete, even for two subclasses of
co-comparability graphs: permutation graphs (for $h \geq
6$)~\cite{Lonc-MES} and interval graphs (for $h \geq
4$)~\cite{B-J-MES2}. The hardness on interval graphs implies the
hardness for graphs of bounded thinness, since interval graphs are
the graphs with thinness $1$.

Both algorithms, the one for maximum weighted stable set and the
one for capacitated coloring with fixed number of colors, are
based on dynamic programming. One of the main results in this work
is a generalization of these algorithmic results. We describe now a framework of problems that can be solved for graphs with bounded thinness, given the representation.
\\


\noindent \textit{Instance:}
\begin{itemize}
\item A $k$-thin representation of $G=(V,E)$, with ordering $<$ of
$V$, namely $v_1 < \dots < v_n$, and partition of $V$ into $k$
classes $V^1, \dots, V^k$.

\item A family of arbitrary nonnegative weights {$w_1,\ldots,w_t$}
on $V$.

\item A family of nonnegative weights {$b_1,\ldots,b_p$} on $V$
bounded by a fixed polynomial in $n$ ($p$ fixed, $q(n)$ the bound
for the weights).
\end{itemize}

\noindent \textit{Question:} find sets {$S_1, \dots, S_r$} ($r$
fixed, not necessarily disjoint), $S_j \subseteq V$ for $1 \leq j
\leq r$, such that:

\begin{itemize}

\item the objective is to minimize or maximize a linear function
\linebreak
 $\sum_{1 \leq i \leq t; 1 \leq j \leq r}
c_{ij}w_i(S_j)$.

\item each vertex $v$ has a list $L(v)$ of combinations of the
sets $S_1, \dots, S_r$ to which it can belong (that may include
the empty combination).

\item there is an $r\times r$ symmetric matrix $M$ over $0, 1, *$,
stating the adjacency conditions on the sets $S_j$, such that for
$1 \leq i < j \leq r$, $M_{ii} = 1$ means $S_i$ is a clique,
$M_{ii} = 0$ means $S_i$ is a stable set, $M_{ij} = 1$ means all
the edges joining $S_i$ and $S_j$ have to be present, $M_{ij} = 0$
means there are no edges from $S_i$ to $S_j$.

\item there is a family of restrictions on the weight of the intersection and of the union of some families of sets. Such restrictions can be expressed as

\begin{itemize}
\item $0\leq l_{iJ\cap} \leq b_i(\bigcap_{j \in J} S_j) \leq
u_{iJ\cap}$, such that $1 \leq i \leq p$, $J \subseteq
\{1,\dots,r\}$.

\item $0\leq l_{iJ\cup} \leq b_i(\bigcup_{j \in J} S_j) \leq
u_{iJ\cup}$, such that $1 \leq i \leq p$, $J \subseteq
\{1,\dots,r\}$.
\end{itemize}

Notice that some of these restrictions can be of cardinality, if the corresponding weight function $b_i$ is constant.
\end{itemize}

The family of problems that can be modeled within this framework
includes weighted variations of list matrix partition problems
with matrices of bounded size, which in turn generalize coloring,
list coloring, list homomorphism, equitable coloring with
different objective functions, all for fixed number of colors (or
graph size in the case of homomorphism), clique cover with fixed
number of cliques, weighted stable sets, and other graph partition
problems. It models also sum-coloring and its more general version optimum cost chromatic partition problem~\cite{Jan-color} for fixed number of colors, but it does not include dominating-like problems. \\


We will solve such a problem as a shortest or longest path problem
(according to minimization or maximization of the objective
function) in an auxiliary acyclic digraph $D=(X,A)$ whose nodes
correspond to \emph{states} and whose arcs are weighted and
labeled. The total weight of the path is the value of the
objective function in the solution that can be built by using
the arc labels. We will used the term ``nodes'' for the digraph $D$ in order to avoid confusion with the vertices of the graph $G$.

A state is a tuple, containing:

\begin{itemize}

\item a number $1 \leq s \leq |V_G|$ indicating that we are
considering the subgraph $G_s$ of $G$, induced by $v_1, \dots,
v_s$.

\item nonnegative parameters $l_{iJ\cap}, u_{iJ\cap}, l_{iJ\cup},
u_{iJ\cup}$, for $1 \leq i \leq p$, $J \subseteq \{1,\dots,r\}$;
they are at most $2^{r+2}p$, and each of them may take a
nonnegative value at most $nq(n)$, which is an upper bound for
$b_i(V)$, for every $1 \leq i \leq p$.

\item a family of nonnegative parameters $\{\alpha_{ij}\}_{1 \leq
i \leq k; 1 \leq j \leq r}$, meaning that we cannot pick for $S_j$
a vertex of the first $\alpha_{ij}$ vertices of the set $V^i$ of
the partition; there are $kr$ such parameters and each of them may take a
nonnegative value at most $n-1$.

\item a family of nonnegative parameters $\{\beta_{ij}\}_{1 \leq i
\leq k; 1 \leq j \leq r}$, meaning that we cannot pick for $S_j$ a
vertex on the last $\beta_{ij}$ vertices of the set $V^i$ of the
partition; there are $kr$ such parameters  and each of them may take a nonnegative
value at most $n-1$.

\end{itemize}

The total number of states is then at most
$n^{2kr+1}(nq(n))^{2^{r+2}p}$, that is polynomial in $n$, since
$k$, $r$, and $p$ are constant and $q(n)$ is polynomial in $n$.

The digraph $D$ will have nodes that correspond to possible
states, organized in layers $X_0, X_1, \dots, X_n$ such that $X_0$
contains only one node $x_0$, and the layer $X_s$ contains the
states whose first parameter is $s$. The layer $X_n$ contains also
only one node, corresponding to the state $(n,$ $\{l_{iJ\cap}\},$
$\{u_{iJ\cap}\},$ $\{l_{iJ\cup}\},$ $\{u_{iJ\cup}\},$
$\{\alpha_{ij}\},$ $\{\beta_{ij}\})$, where the parameters
$\{l_{iJ\cap}\}, \{u_{iJ\cap}\}, \{l_{iJ\cup}\}, \{u_{iJ\cup}\}$
are the ones in the original formulation of the problem and
$\alpha_{ij} = \beta_{ij} = 0$ for every $1 \leq i \leq k$, $1
\leq j \leq r$.

All arcs of $A$ have the form $(u, w)$ with $u\in X_s$ and $w\in
X_{s+1}$, for some $0\leq s \leq n-1$.

We associate with each node of $X$ a suitable problem, in the same
framework, whose parameters correspond to the parameters in the
state, but with additional constraints associated with the
parameters $\{\alpha_{ij}\}$ and $\{\beta_{ij}\}$.

We will define the arcs in such a way that a node is reachable
from the node in the layer $X_0$ if and only if the associated
problem has a solution. The length of the path will be the weight
of the solution, and the set of arc labels will encode the
solution. Let us describe the arcs of the digraph. \\

Let $w$ be a node with parameters $(1,$ $\{l_{iJ\cap}\},$
$\{u_{iJ\cap}\},$ $\{l_{iJ\cup}\},$ $\{u_{iJ\cup}\},$
$\{\alpha_{ij}\},$ $\{\beta_{ij}\})$.

Let $1 \leq \ell \leq k$ such that $v_1 \in V^\ell$. For each
$\tilde{J} \in L(v_1)$ (in particular $\tilde{J} \subseteq
\{1,\dots,r\}$), such that:

\begin{enumerate}[{1.}1]
\item \label{cond:1} For each $j \in \tilde{J}$, $\beta_{\ell j} =
\alpha_{\ell j} = 0$.

\item \label{cond:2} For each $J \subseteq \tilde{J}$, $l_{iJ\cap}
\leq b_i(v_1) \leq u_{iJ\cap}$.

\item \label{cond:3} For each $J \not \subseteq \tilde{J}$,
$l_{iJ\cap} = 0$.

\item \label{cond:4} For each $J$ such that $J \cap \tilde{J} \neq
\emptyset$, $l_{iJ\cup} \leq b_i(v_1) \leq u_{iJ\cup}$.

\item \label{cond:5} For each $J$ such that $J \cap \tilde{J} =
\emptyset$, $l_{iJ\cup} = 0$.
\end{enumerate}

we add an arc from $x_0$ to $w$, labeled by $\tilde{J}$ and of
weight $\sum_{1 \leq i \leq t; j \in \tilde{J}} c_{ij}w_i(v_1)$.
If no $\tilde{J}$ satisfies conditions
1.\ref{cond:1}--1.\ref{cond:5}, no arc ending in $w$ is added. If
more than one arc $x_0w$ was added, we can keep only the one with
maximum (resp. minimum) weight if we are solving a maximization
(resp. minimization) problem.

Note that if we add the arc $x_0w$ labeled by $\tilde{J}$, then
the solution $S_j = \{v_1\}$ for $j \in \tilde{J}$, $S_j =
\emptyset$ for $j \not \in \tilde{J}$ has weight $\sum_{1 \leq i
\leq t; j \in \tilde{J}} c_{ij}w_i(v_1)$ and satisfies the state
described by $w$: condition 1.\ref{cond:1} says that $v_1$ (the
first and last vertex of $V^\ell$ in $G_1$) is allowed to be
picked for every set $S_j$ for $j \in \tilde{J}$; conditions
1.\ref{cond:2}--1.\ref{cond:5} say that the assignment does not
violate weight constraints. \\

Let $w$ be a node with parameters $(s,$ $\{l_{iJ\cap}\},$
$\{u_{iJ\cap}\},$ $\{l_{iJ\cup}\},$ $\{u_{iJ\cup}\},$
$\{\alpha_{ij}\},$ $\{\beta_{ij}\})$, $1 < s \leq n$.

Let $1 \leq \ell \leq k$ such that $v_s \in V^\ell$. For each
$\tilde{J} \in L(v_s)$, such that:

\begin{enumerate}[{$s$.}1]
\item \label{cond:s1} For each $j \in \tilde{J}$, $\beta_{\ell j}
= 0$.

\item  \label{cond:s2} For each $j \in \tilde{J}$, $\alpha_{\ell
j} < |V^\ell \cap \{v_1, \dots, v_s\}|$.

\item \label{cond:s3} For each $J \subseteq \tilde{J}$, $b_i(v_s)
\leq u_{iJ\cap}$.

\item \label{cond:s4} For each $J$ such that $J \cap \tilde{J}
\neq \emptyset$, $b_i(v_s) \leq u_{iJ\cup}$.
\end{enumerate}

we add an arc from $u$ to $w$, labeled by $\tilde{J}$ and of
weight $\sum_{1 \leq i \leq t; j \in \tilde{J}} c_{ij}w_i(v_s)$,
where $u$ has parameters $(s-1,$ $\{l'_{iJ\cap}\},$
$\{u'_{iJ\cap}\},$ $\{l'_{iJ\cup}\},$ $\{u'_{iJ\cup}\},$
$\{\alpha'_{ij}\},$ $\{\beta'_{ij}\})$, such that:

\begin{enumerate}[{$s'$.}1]
\item \label{conds1:1} Let $1 \leq j \leq r$. If there exists $j'
\in \tilde{J}$ such that $M_{jj'} = 0$, then $\beta'_{\ell j} =
\max\{\beta_{\ell j}-1,|N(v_s) \cap V^\ell \cap \{1, \dots,
s-1\}|\}$, and for $1 \leq i \leq k$, $i\neq \ell$, $\beta'_{ij} =
\max\{\beta_{ij},|N(v_s) \cap V^i \cap \{1, \dots, s-1\}|\}$.
Otherwise, $\beta'_{\ell j} = \max\{0,\beta_{\ell j}-1\}$, and for
$1 \leq i \leq k$, $i\neq \ell$, $\beta'_{ij} = \beta_{ij}$.

\item \label{conds1:2} Let $1 \leq j \leq r$. If there exists $j'
\in \tilde{J}$ such that $M_{jj'} = 1$, then

$\alpha'_{\ell j} = \max\{\min\{|V^\ell \cap \{1, \dots,
s-1\}|,\alpha_{\ell j}\},|\overline{N}(v_s) \cap V^\ell \cap \{1,
\dots, s-1\}|\}$, and for $1 \leq i \leq k$, $i\neq \ell$,
$\alpha'_{ij} = \max\{\alpha_{ij},|\overline{N}(v_s) \cap V^i \cap
\{1, \dots, s-1\}|\}$. Otherwise, $\alpha'_{\ell j} =
\min\{|V^\ell \cap \{1, \dots, s-1\}|,\alpha_{\ell j}\}$, and for
$1 \leq i \leq k$, $i\neq \ell$, $\alpha'_{ij} = \alpha_{ij}$.

\item \label{conds1:3} For each $J \subseteq \tilde{J}$,
$l'_{iJ\cap} = \max\{0,l_{iJ\cap} - b_i(v_s)\}$ and $u'_{iJ\cap} =
u_{iJ\cap} - b_i(v_s)$.

\item \label{conds1:4} For each $J \not \subseteq \tilde{J}$,
$l'_{iJ\cap} = l_{iJ\cap}$ and $u'_{iJ\cap} = u_{iJ\cap}$.

\item \label{conds1:5} For each $J$ such that $J \cap \tilde{J}
\neq \emptyset$, $l'_{iJ\cup} = \max\{0,l_{iJ\cup} - b_i(v_s)\}$
and $u'_{iJ\cup} = u_{iJ\cup} - b_i(v_s)$.

\item \label{conds1:6} For each $J$ such that $J \cup \tilde{J} =
\emptyset$, $l'_{iJ\cup} = l_{iJ\cup}$ and $u'_{iJ\cup} =
u_{iJ\cup}$.
\end{enumerate}

If no $\tilde{J}$ satisfies conditions
$s$.\ref{cond:s1}--$s$.\ref{cond:s4}, no arc ending in $w$ is
added. If more than one arc from the same vertex $u$ to $w$ was
added, we can keep only the one with maximum (resp. minimum)
weight if we are solving a maximization (resp. minimization)
problem.

That is, if an arc is added, the arc corresponds to the choice of
adding the vertex $v_s$ to the sets $\{S_j\}_{j\in \tilde{J}}$,
the conditions required imply that the choice is valid for $w$ in
the case that the state described by $u$ admits a solution, the
label of the arc keeps track of the choice made, and the cost
corresponds to the weight that the choice adds to the objective
function.

Note that if we add the arc $uw$ labeled by $\tilde{J}$, then for
a solution $\{S'_j\}_{1\leq j\leq r}$ for $G_{s-1}$ satisfying the
state described by $u$, then the solution $\{S_j\}_{1\leq j\leq
r}$ for $G_{s}$ such that $S_j = S'_j \cup \{v_s\}$ for $j\in
\tilde{J}$, $S_j = S'_j$ for $j\not \in \tilde{J}$ satisfies the
state described by $w$. Conditions $s$.\ref{cond:s1} and
$s$.\ref{cond:s2} say that $v_s$ (the last vertex of $V^\ell$ in
$G_{s}$) is allowed to be picked for every set  $S_j$ for $j \in
\tilde{J}$. Condition $s'$.\ref{conds1:1} ensures on one hand that
the conditions imposed by the parameters $\{\beta_{ij}\}$ in $w$
are satisfied by the solution of $u$, and, on the other hand, that
if $j' \in \tilde{J}$ and $1 \leq j \leq r$ are such that
$M_{jj'}=0$ then no neighbor of $v_s$ belongs to $S'_j$, as
required. Similarly, condition $s'$.\ref{conds1:2} ensures on one
hand that the conditions imposed by the parameters
$\{\alpha_{ij}\}$ in $w$ are satisfied also by the solution of
$u$, and, on the other hand, that if $j' \in \tilde{J}$ and $1
\leq j \leq r$ are such that $M_{jj'}=1$ then all vertices in
 $S'_j$ are adjacent to $v_s$, as required. These conditions
 strongly use that the order and the partition are consistent.
Finally, conditions $s$.\ref{conds1:3}--$s$.\ref{conds1:4}, and
$s'$.\ref{conds1:3}--$s'$.\ref{conds1:6} ensure that the solution
does not violate weight constraints.

Moreover, the difference of weight of the solution $\{S_j\}_{1\leq
j\leq r}$ with respect to $\{S'_j\}_{1\leq j\leq r}$ is exactly
$\sum_{1
\leq i \leq t; j \in \tilde{J}} c_{ij}w_i(v_s)$.\\

In that way, a directed path in the digraph corresponds to an
assignment of vertices of the graph to lists of sets and its
weight is the value of the objective function for the
corresponding assignment.

The digraph has a polynomial number of nodes and can be built in
polynomial time. Since it is acyclic, both the longest path and
shortest path can be computed in linear time in the size of the
digraph by topological sorting.

\begin{remark}
The thinness is not preserved by the complement operation of graphs (see for instance Theorem~\ref{thm:tK2}). However, for every fixed $k$, all the problems that can be modeled in this framework can be solved for the complement $\overline{G}$ of a $k$-thin graph $G$, in the same framework, simply by swapping ones and zeroes in the restriction matrix $M$.
\end{remark}

\subsection{Extending the family of combinatorial optimization problems solvable on graphs with bounded proper
thinness}\label{sec:algo-pthin-comb}

\bigskip

We start by the following observation: in a proper $k$-thin
representation of a graph $G$, with ordering $<$ of $V$, namely
$v_1 < \dots < v_n$, and partition of $V$ into $k$ classes $V^1,
\dots, V^k$, for each pair of vertices $v_s < v_r$ that are in the
same class, $N[v_s] \cap \{v_1, \dots, v_s\} \supseteq N[v_r] \cap
\{v_1, \dots, v_s\}$. This allows us to handle other kinds of
restrictions as for example domination type constraints.

Namely, if we are considering the subgraph $G_s$ of $G$ induced by
$\{v_1, \dots, v_s\}$ but we ``keep in mind'' that we still need
to dominate some of the vertices in $\{v_{s+1}, \dots, v_n\}$ with
vertices of $G_s$, we can summarize these conditions into at most
$k$ of them (each imposed by vertices of $\{v_{s+1}, \dots, v_n\}$
in each partition class).

For graphs with bounded proper thinness $k$, when the proper
$k$-thin representation of the graph is given, we can add now to
the instance (with respect to Section~\ref{sec:algo-thin-comb})
this kind of restrictions:

\begin{itemize}
\item {$l_{ij(N)} \leq |S_i \cap N(v)| \leq u_{ij(N)} \quad
\forall v \in S_j$}, such that $l_{ij(N)} \in \{0,1\}$ and
$u_{ij(N)} \in \{1,\infty\}$ (it can be $i=j$), $1 \leq i,j \leq
r$.

\item {$l_{ij[N]} \leq |S_i \cap N[v]| \leq u_{ij[N]} \quad
\forall v \in S_j$}, such that $l_{ij[N]}\in \{0,1\}$ and
$u_{ij[N]}\in \{1,\infty\}$ (it can be $i=j$), $1 \leq i,j \leq
r$.
\end{itemize}

In this way the framework includes {domination-type} problems in
the literature and their {weighted} versions, such as
existence/minimum (weighted) {independent dominating set}, minimum
(weighted) {total} dominating set, minimum {perfect} dominating
set and existence/minimum (weighted) {efficient} dominating set,
b-coloring~\cite{I-M-b-col} with fixed number of colors. \\


We will keep the notation of Section~\ref{sec:algo-thin-comb} and
describe how to modify the algorithm in order to take into account
the new restrictions. Now the vertex order and the partition of
$G$ are strongly consistent.

Each state now will be augmented with some new parameters:

\begin{itemize}
\item a family of nonnegative parameters $\{\gamma_{ij}\}_{1 \leq
i \leq k; 1 \leq j \leq r}$, meaning that the last $\gamma_{ij}$
vertices of $V^i$ have already a neighbor in $S_j$ (of index
higher than them); there are $kr$ such parameters and each of them may take a
nonnegative value at most $n-1$.

\item a family of nonnegative parameters $\{\gamma^2_{ij}\}_{1
\leq i \leq k; 1 \leq j \leq r}$, meaning that the last
$\gamma^2_{ij}$ vertices of $V^i$ have already two neighbors in
$S_j$ (of index higher than them); there are $kr$ such parameters  and each of them
may take a nonnegative value at most $n-1$.

\item a family of nonnegative parameters $\{\lambda_{ijc}\}_{1
\leq i,c \leq k; 1 \leq j \leq r}$, meaning that, for each value
$1 \leq c \leq k$, $S_j$ has to contain at least one vertex in the
set that is the union over $1 \leq i \leq k$ of the last
$\lambda_{ijc}$ vertices of $V^i$ (if the union is empty, this means no restriction associated with $(c,S_j)$); there are $k^2r$ such parameters
and each of them may take a nonnegative value at most $n-1$.
\end{itemize}

The total number of states is then multiplied by at most
$n^{k^2r+2kr}$, that keeps it polynomial in $n$, since $k$ and
$r$ are constant.\\

The value of all these parameters in the only node of the layer
$X_n$ of the digraph is zero. \\

Now the problems associated with the nodes of $X$ will have the
additional constraints associated with the new restrictions and
the parameters $\{\gamma_{ij}\}$, $\{\gamma^2_{ij}\}$, and
$\{\lambda_{ijc}\}$.

Let us describe the additional conditions for the arcs of the digraph, whose labels and weights are still the same as in Section~\ref{sec:algo-thin-comb}. \\

Let $w$ be a node with parameters $(1,$ $\dots,$
$\{\gamma_{ij}\},$ $\{\gamma^2_{ij}\},$ $\{\lambda_{ijc}\})$.

Let $1 \leq \ell \leq k$ such that $v_1 \in V^\ell$. For each
$\tilde{J} \in L(v_1)$ (in particular $\tilde{J} \subseteq
\{1,\dots,r\}$) satisfying 1.\ref{cond:1}--1.\ref{cond:5}, and
such that:

\begin{enumerate}[{1.}1]
\setcounter{enumi}{5}

\item \label{cond:6} For each $1 \leq i \leq r$, $j \in
\tilde{J}$, such that $l_{ij(N)} = 1$, $\gamma_{\ell i} > 0$.

\item \label{cond:7} For each $i \not \in \tilde{J}$, $j \in
\tilde{J}$, such that $l_{ij[N]} = 1$, $\gamma_{\ell i} > 0$.

\item \label{cond:8} For each $1 \leq i \leq r$, $j \in
\tilde{J}$, such that $u_{ij(N)} = 1$ or $u_{ij[N]} = 1$,
$\gamma^2_{\ell i} = 0$.

\item \label{cond:9} For each $i, j \in \tilde{J}$, such that
$u_{ij[N]} = 1$, $\gamma_{\ell i} = 0$.

\item \label{cond:10} For each $j \not \in \tilde{J}$ and for each
$1\leq c \leq k$, $\lambda_{\ell jc} = 0$.
\end{enumerate}

we add an arc from $x_0$ to $w$, labeled by $\tilde{J}$ and of
weight $\sum_{1 \leq i \leq t; j \in \tilde{J}} c_{ij}w_i(v_1)$.
If no $\tilde{J}$ satisfies conditions
1.\ref{cond:1}--1.\ref{cond:10}, no arc ending in $w$ is added. If
more than one arc $x_0w$ was added, we can keep only the one with
maximum (resp. minimum) weight if we are solving a maximization
(resp. minimization) problem.

Note that if we add the arc $x_0w$ labeled by $\tilde{J}$, then
the solution $S_j = \{v_1\}$ for $j \in \tilde{J}$, $S_j =
\emptyset$ for $j \not \in \tilde{J}$ has weight $\sum_{1 \leq i
\leq t; j \in \tilde{J}} c_{ij}w_i(v_1)$ and satisfies the state
described by $w$: conditions 1.\ref{cond:1}--1.\ref{cond:5} ensure
the properties required in Section~\ref{sec:algo-thin-comb}; conditions 1.\ref{cond:6}--1.\ref{cond:9} ensure the validity of the two new families of restrictions about lower and upper bounds of
neighbors of vertices of one set in other set, and condition 1.\ref{cond:10} ensures that the restrictions imposed by the parameters $\{\lambda_{ijc}\}$ are satisfied. \\

Let $w$ be a node with parameters $(s,$ $\{l_{iJ\cap}\},$
$\{u_{iJ\cap}\},$ $\{l_{iJ\cup}\},$ $\{u_{iJ\cup}\},$
$\{\alpha_{ij}\},$ $\{\beta_{ij}\}$, $\{\gamma_{ij}\},$
$\{\gamma^2_{ij}\},$ $\{\lambda_{ijc}\})$, $1 < s \leq n$.

Let $1 \leq \ell \leq k$ such that $v_s \in V^\ell$. For each
$\tilde{J} \in L(v_s)$ satisfying
$s$.\ref{cond:s1}--$s$.\ref{cond:s4}, and such that:

\begin{enumerate}[{$s$.}1]
\setcounter{enumi}{4}

\item \label{cond:s5} For each $1 \leq i \leq r$, $j \in
\tilde{J}$, such that $u_{ij(N)} = 1$ or $u_{ij[N]} = 1$,
$\gamma^2_{\ell i} = 0$.

\item \label{cond:s6} For each $i, j \in \tilde{J}$, such that
$u_{ij[N]} = 1$, $\gamma_{\ell i} = 0$.

\item \label{cond:s7} For each $j \not \in \tilde{J}$ and for each
$1\leq c \leq k$, either $\lambda_{\ell jc} = 0$, or
$\lambda_{\ell jc} > 1$, or there exists $1 \leq i \leq k$, $i
\neq \ell$, such that $\lambda_{ijc} > 0$ (i.e., the union over $1
\leq i \leq k$ of the last $\lambda_{ijc}$ vertices of $V^i$ is
not $\{v_s\}$).

\item \label{cond:s8} For each $1 \leq i \leq r$ such that
$\gamma_{\ell i} = 0$ and there exists $j \in \tilde{J}$ such that
$l_{ij(N)} = 1$, $N(v_s) \cap \{1, \dots, s-1\} \neq \emptyset$.

\item \label{cond:s9} For each $i \not \in \tilde{J}$ such that
$\gamma_{\ell i} = 0$ and there exists $j \in \tilde{J}$ such that
$l_{ij[N]} = 1$, $N(v_s) \cap \{1, \dots, s-1\} \neq \emptyset$.
\end{enumerate}

Let $\{\lambda^0_{ijc}\}_{1 \leq i,c\leq k; 1 \leq j \leq r}$ be
defined this way: for every $j\in \tilde{J}$ and every $1 \leq c
\leq k$ such that $\lambda_{\ell jc} > 0$, let $\lambda^0_{ijc} =
0$ for every $1 \leq i \leq k$; for every $j \in \tilde{J}$ and
every $1 \leq c \leq k$ such that $\lambda_{\ell jc} = 0$, let
$\lambda^0_{ijc} = \lambda_{ijc}$ for every $1 \leq i \leq k$; for
every $j\not \in \tilde{J}$ and every $1 \leq c \leq k$, let
$\lambda^0_{\ell jc} = \max\{0,\lambda_{\ell jc}-1\}$ and let
$\lambda^0_{ijc} = \lambda_{ijc}$ for every $1 \leq i \leq k$, $i
\neq \ell$.

Let $\{\lambda^1_{jc}\}_{1 \leq c\leq k; 1 \leq j \leq r}$ be
defined as $\lambda^1_{jc} = 0$ if $\lambda^0_{ijc} = 0$ for every
$1 \leq i \leq k$, $\lambda^1_{jc} = 1$ otherwise.

We add an arc from $u$ to $w$, labeled by $\tilde{J}$ and of
weight $\sum_{1 \leq i \leq t; j \in \tilde{J}} c_{ij}w_i(v_s)$,
where $u$ has parameters $(s-1,$ $\{l'_{iJ\cap}\},$
$\{u'_{iJ\cap}\},$ $\{l'_{iJ\cup}\},$ $\{u'_{iJ\cup}\},$
$\{\alpha'_{ij}\},$ $\{\beta'_{ij}\}$, $\{\gamma'_{ij}\},$
$\{{\gamma^2}'_{ij}\},$ $\{\lambda'_{ijc}\})$, satisfies
conditions $s'$.\ref{conds1:2}--$s'$.\ref{conds1:6}, and:

\begin{enumerate}[{$s'$.}1]
\setcounter{enumi}{6}

\item \label{conds1:7} For each $1 \leq i \leq r$ such that
$\gamma_{\ell i} = 0$ and there exists $j \in \tilde{J}$ such that
$l_{ij(N)} = 1$, if $\lambda^1_{i\ell} = 0$, then
$\lambda'_{j'i\ell} = |N(v_s) \cap V^{j'} \cap \{1, \dots, s-1\}|$
for each $1 \leq j' \leq k$; otherwise, $\lambda'_{j'i\ell} =
\lambda^0_{j'i\ell}$ for every $1 \leq j' \leq k$ (recall that, by
the observations above about proper thinness, $\lambda^0_{j'i\ell}
= \min\{\lambda^0_{j'i\ell},|N(v_s) \cap V^{j'} \cap \{1, \dots,
s-1\}|\}$).

\item \label{conds1:8} For each $i \not \in \tilde{J}$ such that
$\gamma_{\ell i} = 0$ and there exists $j \in \tilde{J}$ such that
$l_{ij[N]} = 1$, if $\lambda^1_{i\ell} = 0$, then
$\lambda'_{j'i\ell} = |N(v_s) \cap V^{j'} \cap \{1, \dots, s-1\}|$
for each $1 \leq j' \leq k$; otherwise, $\lambda'_{j'i\ell} =
\lambda^0_{j'i\ell}$ for every $1 \leq j' \leq k$.

\item \label{conds1:9} For each $i, j, c$ not comprised in
conditions $s'$.\ref{conds1:7} and $s'$.\ref{conds1:8},
$\lambda'_{ijc} = \lambda^0_{ijc}$.

\item \label{conds1:10} Let $1 \leq j \leq r$. If there exists $j'
\in \tilde{J}$ satisfying at least one of the following:
\begin{itemize}
\item $M_{jj'} = 0$

\item ($u_{jj'(N)} = 1$ or $u_{jj'[N]} = 1$) and $\gamma_{\ell j}
> 0$

\item $j \in \tilde{J}$ and $u_{jj'[N]} = 1$
\end{itemize}
then, $\beta'_{\ell j} = \max\{\beta_{\ell j}-1,|N(v_s) \cap
V^\ell \cap \{1, \dots, s-1\}|\}$, and for $1 \leq i \leq k$,
$i\neq \ell$, $\beta'_{ij} = \max\{\beta_{ij},|N(v_s) \cap V^i
\cap \{1, \dots, s-1\}|\}$. Otherwise, $\beta'_{\ell j} =
\max\{0,\beta_{\ell j}-1\}$, and for $1 \leq i \leq k$, $i\neq
\ell$, $\beta'_{ij} = \beta_{ij}$.

\item \label{conds1:11} For each $j \in \tilde{J}$: if $|N(v_s)
\cap V^{\ell} \cap \{1, \dots, s-1\}| \geq \gamma_{\ell j} - 1$,
then $\gamma'_{\ell j} = |N(v_s) \cap V^{\ell} \cap \{1, \dots,
s-1\}|$ and ${\gamma^2}'_{\ell j} = \max\{0,\gamma_{\ell j}-1\}$;
otherwise, $\gamma'_{\ell j} = \max\{0,\gamma_{\ell j}-1\}$ and
${\gamma^2}'_{\ell j} = \max\{\gamma^2_{\ell j}-1,|N(v_s) \cap
V^{\ell} \cap \{1, \dots, s-1\}|\}$.

\item \label{conds1:12} For each $j \in \tilde{J}$, $1 \leq i \leq
k$, $i \neq \ell$: if $|N(v_s) \cap V^{i} \cap \{1, \dots, s-1\}|
\geq \gamma_{ij}$, then $\gamma'_{ij} = |N(v_s) \cap V^{i} \cap
\{1, \dots, s-1\}|$ and ${\gamma^2}'_{ij} = \gamma_{ij}$;
otherwise, $\gamma'_{ij} = \gamma_{ij}$ and ${\gamma^2}'_{ij} =
\max\{\gamma^2_{ij},|N(v_s) \cap V^{i} \cap \{1, \dots, s-1\}|\}$.
\end{enumerate}

If no $\tilde{J}$ satisfies conditions
$s$.\ref{cond:s1}--$s$.\ref{cond:s9}, no arc ending in $w$ is
added. If more than one arc from the same vertex $u$ to $w$ was
added, we can keep only the one with maximum (resp. minimum)
weight if we are solving a maximization (resp. minimization)
problem.

That is, if an arc is added, the arc corresponds to the choice of
adding the vertex $v_s$ to the sets $\{S_j\}_{j\in \tilde{J}}$,
the conditions required imply that the choice is valid for $w$ in
the case that the state described by $u$ admits a solution, the
label of the arc keeps track of the choice made, and the cost
corresponds to the weight that the choice adds to the objective
function.

Note that if we add the arc $uw$ labeled by $\tilde{J}$, then for
a solution $\{S'_j\}_{1\leq j\leq r}$ for $G_{s-1}$ satisfying the
state described by $u$, then the solution $\{S_j\}_{1\leq j\leq
r}$ for $G_{s}$ such that $S_j = S'_j \cup \{v_s\}$ for $j\in
\tilde{J}$, $S_j = S'_j$ for $j\not \in \tilde{J}$ satisfies the
state described by $w$.

Condition $s'$.\ref{conds1:10} ensures on one hand that the
conditions imposed by the parameters $\{\beta_{ij}, u_{ij(N)},
u_{ij[N]}\}$ in $w$ are satisfied by the solution of $u$, and, on
the other hand, that if $j' \in \tilde{J}$ and $1 \leq j \leq r$
are such that $M_{jj'}=0$ then no neighbor of $v_s$ belongs to
$S'_j$, as required. Conditions
$s'$.\ref{conds1:7}--$s'$.\ref{conds1:9} together with
$s$.\ref{cond:s7}--$s$.\ref{cond:s9} define parameters
$\{\lambda'_{ijc}\}$ in $u$ in order to guarantee in $w$ both the
conditions imposed by the lower bounds $\{l_{ij(N)},l_{ij[N]}\}$
and those imposed by the parameters $\{\lambda_{ijc}\}$. Finally,
conditions $s'$.\ref{conds1:11} and $s'$.\ref{conds1:12} properly
update the definition of parameters
$\{\gamma'_{ij},{\gamma^2}'_{ij}\}$ according to the choice
$\tilde{J}$ for $v_s$. Conditions
$s'$.\ref{conds1:2}--$s'$.\ref{conds1:6} were analyzed above in
Section~\ref{sec:algo-thin-comb}.

As in that case, the difference of weight of the solution
$\{S_j\}_{1\leq j\leq r}$ with respect to $\{S'_j\}_{1\leq j\leq
r}$ is exactly $\sum_{1
\leq i \leq t; j \in \tilde{J}} c_{ij}w_i(v_s)$.\\

In that way, a directed path in the digraph corresponds to an
assignment of vertices of the graph to lists of sets and its
weight is the value of the objective function for the
corresponding assignment.

The digraph has a polynomial number of nodes and can be built in
polynomial time. Since it is acyclic, both the longest path and
shortest path can be computed in linear time in the size of the
digraph by topological sorting.


\section{Thinness and graph operations}\label{sec:thin-and-oper}

In this section we analyze the behavior of the thinness and proper
thinness under different graph operations, namely union, join, and
Cartesian product. The first two results allow us to fully
characterize $k$-thin graphs by forbidden induced subgraphs within
the class of cographs. The third result is used to solve in
polynomial time the $t$-rainbow domination problem for fixed $t$
on graphs with
bounded thinness. \\

Let $G_1 = (V_1,E_1)$ and $G_2 = (V_2,E_2)$ be two graphs with
$V_1 \cap V_2 = \emptyset$. The \emph{union} of $G_1$ and $G_2$ is
the graph $G_1 \cup G_2 = (V_1 \cup V_2, E_1 \cup E_2)$, and the
\textit{join} of $G_1$ and $G_2$ is the graph $G_1 \vee G_2 = (V_1
\cup V_2, E_1 \cup E_2 \cup V_1 \times V_2)$ (i.e.,
$\overline{G_1\vee G_2} = \overline{G_1} \cup \overline{G_2}$).

\begin{theorem}\label{thm:union}
Let $G_1$ and $G_2$ be graphs. Then $\thin(G_1 \cup G_2) =
\max\{\thin(G_1),$ $\thin(G_2)\}$ and $\pthin(G_1 \cup G_2) =
\max\{\pthin(G_1),\pthin(G_2)\}$.
\end{theorem}

\begin{proof}
Since both  $G_1$ and $G_2$ are induced subgraphs of $G_1 \cup
G_2$, then $\thin(G_1 \cup G_2) \geq
\max\{\thin(G_1),\thin(G_2)\}$ and the same holds for the proper
thinness.

Let $G_1$ and $G_2$ be two graphs with thinness (resp. proper
thinness) $t_1$ and $t_2$, respectively. Let $v_1,\dots,  v_{n_1}$
and $(V_1^1, \ldots, V_1^{t_1})$ be an ordering and a partition of
$V(G_1)$ which are consistent (resp. strongly consistent). Let
$w_1,\dots, w_{n_2}$ and $(V_2^1, \ldots, V_2^{t_2})$ be an
ordering and a partition of $V(G_2)$ which are consistent (resp.
strongly consistent). Suppose without loss of generality that $t_1
\leq t_2$. For $G = G_1 \cup G_2$, define a partition $V^1,\dots,
V^{t_2}$ such that $V^i = V_1^i \cup V_2^i$ for $i = 1, \dots,
t_1$ and $V^i = V_2^i$ for $i = t_1+1, \dots, t_2$, and define
$v_1,\dots, v_{n_1},w_1,\dots, w_{n_2}$ as an ordering of the
vertices. By definition of union of graphs, if three ordered
vertices according to the order defined in $V(G_1 \cup G_2)$ are
such that the first and the third are adjacent, either the three
vertices belong to $V(G_1)$ or the three vertices belong to
$V(G_2)$. Since the order and the partition restricted to each of
$G_1$ and $G_2$ are the original ones, the properties required for
consistency (resp. strong consistency) are satisfied.
\end{proof}

\begin{theorem}\label{thm:join}
Let $G_1$ and $G_2$ be graphs. Then $\thin(G_1 \vee G_2) \leq
\thin(G_1)+\thin(G_2)$ and $\pthin(G_1 \vee G_2) \leq
\pthin(G_1)+\pthin(G_2)$. Moreover, if $G_2$ is complete, then
$\thin(G_1 \vee G_2) = \thin(G_1)$.
\end{theorem}

\begin{proof}
Let $G_1$ and $G_2$ be two graphs with thinness (resp. proper
thinness) $t_1$ and $t_2$, respectively. Let $v_1,\dots,  v_{n_1}$
and $(V_1^1, \ldots, V_1^{t_1})$ be an ordering and a partition of
$V(G_1)$ which are consistent (resp. strongly consistent). Let
$w_1,\dots, w_{n_2}$ and $(V_2^1, \ldots, V_2^{t_2})$ be an
ordering and a partition of $V(G_2)$ which are consistent (resp.
strongly consistent). For $G = G_1 \vee G_2$, define a partition
with $t_1+t_2$ sets as the union of the two partitions, and
$v_1,\dots, v_{n_1},w_1,\dots, w_{n_2}$ as an ordering of the
vertices.

Let $x,y,z$ be three vertices of $V(G)$ such that $x < y < z$, $xz
\in E(G)$, and $x$ and $y$ are in the same class of the partition
of $V(G)$. Then, in particular, $x$ and $y$ both belong either to
$V(G_1)$ or to $V(G_2)$. If $z$ belongs to the same graph, then
$yz \in E(G)$ because the ordering and partition restricted to
each of $G_1$ and $G_2$ are consistent. Otherwise, $z$ is also
adjacent to $y$ by the definition of join.

We have proved that the defined partition and ordering are
consistent, and thus that $\thin(G_1 \vee G_2) \leq
\thin(G_1)+\thin(G_2)$. The proof of the strong consistency, given
the strong consistency of the partition and ordering of each of
$G_1$ and $G_2$, is symmetric and implies $\pthin(G_1 \vee G_2)
\leq \pthin(G_1)+\pthin(G_2)$.

Suppose now that $G_2$ is complete (in particular, $t_2 = 1$).
Since $G_1$ is an induced subgraph of $G_1 \vee G_2$, then
$\thin(G_1 \vee G_2) \geq \thin(G_1)$. For $G = G_1 \vee G_2$,
define a partition $V^1,\dots, V^{t_1}$ such that $V^1 = V_1^1
\cup V_2^1$ and $V^i = V_1^i$ for $i = 2, \dots, t_1$, and define
$v_1,\dots, v_{n_1},w_1,\dots, w_{n_2}$ as an ordering of the
vertices.

Let $x,y,z$ be three vertices of $V(G)$ such that $x < y < z$, $xz
\in E(G)$, and $x$ and $y$ are in the same class of the partition
of $V(G)$. If $z$ belongs to $V(G_2)$, then $z$ is also adjacent
to $y$, because it is adjacent to every vertex in $G-z$. If $z$
belongs to $V(G_1)$, then $x$, $y$, and $z$, belong to $V(G_1)$
due to the definition of the order of the vertices, and thus $yz
\in E(G)$ because the ordering and partition restricted to $G_1$
are consistent. This proves $\thin(G_1 \vee G_2) \leq \thin(G_1)$,
and therefore $\thin(G_1 \vee G_2) = \thin(G_1)$.
\end{proof}

The following lemma shows a way of obtaining graphs with high
thinness, using the join operator.

\begin{lemma}
If $G$ is not complete, then $\thin(G \vee 2K_1) = \thin(G)+1$.
\end{lemma}

\begin{proof}
By Theorem~\ref{thm:join}, $\thin(G \vee 2K_1) \leq
\thin(G)+\thin(2K_1) = \thin(G)+1$. On the other hand, as $G \vee
2K_1$ contains $G$ as induced subgraph, $\thin(G \vee 2K_1) \geq
\thin(G)$.

First notice that if $\thin(G) = 1$ but $G$ is not complete, then
$G \vee 2K_1$ contains $C_4$ as induced subgraph, so it is not an
interval graph, and $\thin(G \vee 2K_1) \geq 2$, as claimed.

Suppose then that $\thin(G) = k > 1$ and $\thin(G \vee 2K_1) = k$
as well, and let $<$ be an ordering of the vertices of $G \vee
2K_1$ consistent with a partition $V^1, \dots, V^k$. Let $v, v'$
be the vertices of the graph $2K_1$, and suppose $v < v'$. Without
loss of generality we may assume $v \in V^k$. As $\thin(G)=k$,
$V^k \cap V(G) \neq \emptyset$. Since $v' > v$, $v'$ is
nonadjacent to $v$, and $v'$ is adjacent to all the vertices in
$V^k \cap V(G)$, $v$ has to be the smallest vertex in $V^k$. Let
$z \in V^k \cap V(G)$ and suppose there is a vertex $x > z$ in
$V(G)$. As $x$ is adjacent to $v'$, it is adjacent to $z$ as well.
So, we can define a new order $<'$ on $V(G \vee 2K_1)$ that
preserves the order $<$ in $V^1 \cup V^{k-1} \cup \{v\}$ and such
that the vertices of $V^k - \{v\}$ are the largest. By the
observations above, this order $<'$ is still consistent with the
partition $V^1, \dots, V^k$. But it is also consistent with the
partition ${V^1}', \dots, {V^k}'$ in which ${V^1}' = V^1 \cup V^k
- \{v\}$, ${V^i}' = V^i$ for $1 < i < k$, and ${V^k}' = \{v\}$.
This implies that $\thin(G) < k$, a contradiction that completes
the proof of the theorem.
\end{proof}

Cographs were defined in~\cite{CorneilLerchsStewart81}, where it
was shown that they are exactly the graphs with no induced path of
length four. Cographs admit a full decomposition theorem. Let the
\emph{trivial} graph be the one with one vertex only.

\begin{proposition}{\em \cite{CorneilLerchsStewart81}}\label{prop:decomp-cografos}
Every cograph that is not trivial is either the union or the join
of two smaller cographs.
\end{proposition}

We will use this structural property along with the theorems about
thinness of union and join of graphs to prove the following.

\begin{theorem} Let $G$ be a cograph and $t \geq 1$. Then $G$ has thinness at most $t$ if and only if $G$ contains no
$\overline{(t+1)K_2}$ as induced subgraph. \end{theorem}

\begin{proof} The only if part holds by Theorem~\ref{thm:tK2}, because the class of $k$-thin graph is
hereditary for every $k$.

We will prove the if part by induction on the number of vertices
of the cograph $G$. If $G$ is a trivial graph, then $\thin(G)=1$
and the theorem holds. If $G$ is not trivial, by Proposition
\ref{prop:decomp-cografos}, it is either union or join of two
smaller cographs $G_1$ and $G_2$, with thinness $t_1$ and $t_2$,
respectively.

Suppose first $G = G_1 \cup G_2$. By Theorem \ref{thm:union},
$\thin(G) = \max\{t_1,t_2\}$. If $t_1$ (resp. $t_2$) is greater
than one, then by inductive hypothesis $G_1$ (resp. $G_2$)
contains $\overline{t_1K_2}$ (resp. $\overline{t_2K_2}$) as
induced subgraph, and so does $G$.

Suppose now that $G = G_1 \vee G_2$. If one of them is complete
(suppose without loss of generality $G_2$), then, by Theorem
\ref{thm:join}, $\thin(G) = t_1$. If $t_1$ is greater than one,
then by inductive hypothesis $G_1$ contains $\overline{t_1K_2}$ as
induced subgraph, and so does $G$. If none of them is complete,
then, by that fact and the inductive hypothesis, $G_1$ contains
$\overline{t_1K_2}$ and $G_2$ contains $\overline{t_2K_2}$ as
induced subgraph. As $\overline{t_1K_2} \vee \overline{t_2K_2} =
\overline{(t_1+t_2)K_2}$, $G$ contains $\overline{(t_1+t_2)K_2}$
as induced subgraph, thus $\thin(G) \geq t1+t2$
(Theorem~\ref{thm:tK2}). By Theorem \ref{thm:join}, $\thin(G) \leq
t1+t2$, and therefore  $\thin(G) = t1+t2$. This finishes the proof
of the theorem. \end{proof}

A characterization by minimal forbidden induced subgraphs for
$k$-thin graphs, $k \geq 2$, is open. \\

Let $G_1 = (V_1,E_1)$ and $G_2 = (V_2,E_2)$ be two graphs. The
Cartesian product $G_1 \square G_2$ is a graph whose vertex set is
the Cartesian product $V_1 \times V_2$, and such that two vertices
$(u_1,u_2)$ and $(v_1,v_2)$ are adjacent in $G_1 \square G_2$ if
and only if either $u_1 = v_1$ and $u_2$ is adjacent to $v_2$ in
$G_2$, or $u_2 = v_2$ and $u_1$ is adjacent to $v_1$ in $G_1$.

\begin{theorem}
Let $G_1$ and $G_2$ be graphs. Then {$\thin(G_1 \square G_2) \leq
\thin(G_1)|V(G_2)|$} and {$\pthin(G_1 \square G_2) \leq
\pthin(G_1)|V(G_2)|$}.
\end{theorem}

\begin{proof}
Let $G_1=(V_1, E_1)$ be a $k$-thin (resp. proper $k$-thin) graph,
and let $v_1,\dots, v_{n_1}$ and $(V_1^1, \ldots, V_1^{k})$ be an
ordering and a partition of $V_1$ which are consistent (resp.
strongly consistent). Let $G_2=(V_2,E_2)$, $n_2=|V_2|$, and
$w_1,\dots, w_{n_2}$ an arbitrary ordering of $V_2$. Consider $V_1
\times V_2$ lexicographically ordered with respect to the
orderings of $V_1$ and $V_2$ above. Consider now the partition
$\{V^{i,j}\}_{1\leq i \leq k,\ 1 \leq j \leq n_2}$ such that
$V^{i,j} = \{(v,w_j) : v \in V_1^i\}$ for each $1\leq i \leq k$,
$1 \leq j \leq n_2$. We will show that this ordering and partition
of $V_1 \times V_2$ are consistent (resp. strongly consistent).
Let $(v_p,w_i), (v_q,w_j), (v_r,w_{\ell})$ be three vertices
appearing in that ordering in $V_1 \times V_2$.

\emph{Case 1: $p = q = r$.} In this case, the three vertices are
in different classes, so no restriction has to be satisfied.

\emph{Case 2: $p = q < r$.} In this case, $(v_p,w_i)$ and
$(v_q,w_j)$ are in different classes. So suppose $G_1$ is proper
$k$-thin and $(v_q,w_j), (v_r,w_{\ell})$ belong to the same class,
i.e., $j=\ell$. Vertices $(v_p,w_i)$ and $(v_r,w_{\ell})$ are
adjacent in $G_1 \square G_2$ if and only if $i = \ell$ and
$v_pv_r \in E_1$. But then $(v_p,w_i)=(v_q,w_j)$, a contradiction.

\emph{Case 3: $p < q = r$.} In this case, $(v_q,w_j)$ and
$(v_r,w_{\ell})$ are in different classes. So suppose $G_1$ is
$k$-thin and $(v_p,w_i), (v_q,w_j)$ belong to the same class,
i.e., $i=j$. Vertices $(v_p,w_i)$ and $(v_r,w_{\ell})$ are
adjacent in $G_1 \square G_2$ if and only if $i = \ell$ and
$v_pv_r \in E_1$. But then $(v_r,w_{\ell})=(v_q,w_j)$, a
contradiction.

\emph{Case 4: $p < q < r$.} Suppose first $G_1$ is $k$-thin (resp.
proper $k$-thin) and $(v_p,w_i), (v_q,w_j)$ belong to the same
class, i.e., $i=j$ and $v_p$, $v_q$ belong to the same class in
$G_1$. Vertices $(v_p,w_i)$ and $(v_r,w_{\ell})$ are adjacent in
$G_1 \square G_2$ if and only if $i = \ell$ and $v_pv_r \in E_1$.
But then $j=\ell$ and since the ordering and the partition are
consistent (resp. strongly consistent) in $G_1$, $v_rv_q \in E_1$
and so $(v_r,w_{\ell})$ and $(v_q,w_j)$ are adjacent in $G_1
\square G_2$. Now suppose that $G_1$ is proper $k$-thin and
$(v_q,w_j), (v_r,w_{\ell})$ belong to the same class, i.e.,
$j=\ell$. Vertices $(v_p,w_i)$ and $(v_r,w_{\ell})$ are adjacent
in $G_1 \square G_2$ if and only if $i = \ell$ and $v_pv_r \in
E_1$. But then $i=j$ and since the ordering and the partition are
strongly consistent in $G_1$, $v_pv_q \in E_1$ and so $(v_p,w_i)$
and $(v_q,w_j)$ are adjacent in $G_1 \square G_2$.
\end{proof}

\begin{corollary}\label{cor:p2thin}
If $G$ is (proper) $k$-thin then $G \square K_t$ is (proper)
$kt$-thin. In particular, if $G$ is a (proper) interval graph then
$G \square K_t$ is (proper) $t$-thin.
\end{corollary}

For a graph $G(V,E)$ and an integer $t$, we say that $f$ is a
\emph{$t$-rainbow dominating function} if it assigns to each vertex
$v \in V$ a subset of $\{1,\dots,t\}$ such that $\cup_{u \in N(v)}
f(u)=\{1,\dots,t\}$ for all $v$ with $f(v)=\emptyset$. Consider
the following generalization of the dominating set problem.

\noindent \textsc{$t$-rainbow domination problem}\\
\emph{Instance:} A graph $G=(V,E)$.\\
\emph{Find:} a $t$-rainbow dominating function that minimizes $\sum_{v \in V}|f(v)|$.\\

The $t$-rainbow domination problem is equivalent to minimum
dominating set of {$G
    \square K_t$}~\cite{B-KS-2R-dom}. As a consequence of Corollary~\ref{cor:p2thin}
and the last remark in Section~\ref{sec:width}, it can be solved
in polynomial time on graphs with bounded thinness for fixed
values of $t$. This generalizes the polynomiality for interval
graphs, recently proved by Hon, Kloks, Liu, and Wang
in~\cite{H-K-L-W-r-dom-arx} (the algorithm for $t=2$ is claimed
in~\cite{H-K-L-W-r-dom}). The problem for proper interval graphs
was stated as an open question by Bre{\v{s}}ar and Kraner
{\v{S}}umenjak in~\cite{B-KS-2R-dom}.

\section{Conclusions and open problems}\label{sec:conc}

We described a wide family of combinatorial optimization problems
that can be solved in polynomial time on classes of bounded
thinness and bounded proper thinness. We think that some
restrictions can be further generalized (specially the domination
type ones), with more involved sets of parameters and transition
rules. We tried to keep it as simpler as possible, yet including
many of the classical combinatorial optimization problems in the
literature.

We also proved a number of theoretical results, some of them
related to the recognition problem for the classes, others
relating the concept of thinness and proper thinness to other
known graph parameters, and analyzing their behavior under the
graph operations union, join, and Cartesian product.

Some open problems are the following.

\begin{itemize}

\item Characterize (proper) $k$-thin graphs by {minimal forbidden
induced subgraphs} (or at least within some graph class, we did it
for thinness in cographs).

\item Find sufficient conditions, for instance a family subgraphs
to forbid as induced subgraphs, for a graph to be (proper)
$k$-thin, even if these graphs are not necessarily forbidden
induced subgraphs for (proper) $k$-thin graphs. These kind of
results have been obtained for MIM-width in~\cite{KANG20171}.

\item Study the behavior of thinness under other graph products or
graph operators in general.

\item What is the {complexity} of computing the {thinness/proper
thinness} of a graph? Or deciding if it is at most $k$ for some
fixed values $k$?

\item Can we develop some randomized algorithm to test just a
subset of vertex orderings and obtain with high probability an
approximation of the thinness/proper thinness?

\item Can we improve the complexity of the algorithms that are in XP to FPT? Or
prove a hardness result?

\item Given a partition of the vertex set into a \emph{fixed}
number $k$ of classes, what is the complexity of deciding if there
is a (strongly) consistent order for the vertices w.r.t. that
partition (and finding it)? (We have proved that for an arbitrary
number of classes the problems are NP-complete, and we have solved
in polynomial time the symmetric problem, i.e., given the
ordering, find a minimum (strongly) consistent partition.)

\end{itemize}

\bigskip

\noindent \textbf{Acknowledgements:} This work was partially
supported by UBACyT Grant 20020130100808BA, CONICET PIP
112-201201-00450CO, ANPCyT PICT 2012-1324 and 2015-2218
(Argentina). We want to thank \-Gianpaolo Oriolo for
bringing our attention into the thinness parameter and for helpful
discussions about it, Daniela Saban for the example of an interval
graph with high proper thinness, Nicolas Nisse for pointing out
the relation between proper thinness and other interval graphs
invariants to us, Bo\v{s}tjan Bre\v{s}ar and Martin Milani\v{c}
for pointing out the open problem on rainbow domination to us, and
Yuri Faenza for his valuable help in the improvement of the
manuscript. Last but not least, we want to thank the anonymous referees for their useful comments that help us to improve the paper.



\end{document}